\documentclass{llncs}

\usepackage{graphicx,color}
\usepackage{tikz}
\usepackage{amsmath,mathrsfs,amsfonts,latexsym,amssymb,verbatim}
\usepackage[mathcal]{euscript}

\newcommand{\VL}[1]{#1}
\newcommand{\VC}[1]{}

\newcommand{\incl}{\subseteq}

\newcommand{\couic}[1]{}

\newcommand{\pa}[1]{\left( #1 \right)}

\newcommand{\vt}{V}
\newcommand{\dom}{\textrm{dom}}
\newcommand{\ports}{\pi}
\newcommand{\port}{\!:\!}

\newcommand{\im}{\operatorname{im}}

\newcommand{\tili}[1]{\widetilde{#1}}

\def\N{\mathbb{N}}

\begin{document}
\title{Generalized Cayley Graphs\\ and Cellular Automata over them}

\author{Pablo Arrighi\inst{1,2} \and Simon Martiel\inst{3} \and Vincent Nesme\inst{1}}
\institute{
Universit\'e de Grenoble, LIG, 220 rue de la chimie, 38400 Saint-Martin-d'H\`eres, France\\
\email{parrighi@imag.fr}, \email{vnesme@imag.fr}
\and
Universit\'e de Lyon, LIP, 46 all\'ee d'Italie, 69008 Lyon, France\\
\and
Universit\'e Nice-Sophia Antipolis, I3S, 2000 routes des Lucioles, 06900 Sophia Antipolis, France\\
\email{martiel@i3s.unice.fr}}

\maketitle

\begin{abstract}
Cayley graphs have a number of useful features: the ability to graphically represent finitely generated group elements and their relations; to name all vertices relative to a point; and the fact that they have a well-defined notion of translation. We propose a notion of graph associated to a language, which conserves or generalizes these features. Whereas Cayley graphs are very regular; associated graphs are arbitrary, although of a bounded degree. Moreover, it is well-known that cellular automata can be characterized as the set of translation-invariant continuous functions for a distance on the set of configurations that makes it a compact metric space; this point of view makes it easy to extend their definition from grids to Cayley graphs. Similarly, we extend their definition to these arbitrary, bounded degree, time-varying graphs. The obtained notion of Cellular Automata over generalized Cayley graphs is stable under composition and under inversion.
\VL{\\{\bf Keywords.} {\em Causal Graph Dynamics, Curtis-Hedlund-Lyndon, Dynamical networks, Boolean networks, Generative networks automata, Graph Automata, Graph rewriting automata, Parallel graph transformations, Amalgamated graph transformations, Time-varying graphs, Regge calculus, Local, No-signalling, Reversibility.}}
\end{abstract}

\section*{Introduction}

\noindent {\em Cayley graphs.} Cayley graphs are graphs associated to a finitely generated group, more precisely to a finite set of generators and their inverses. For instance let this set be $\pi=\{a,a^{-1},b,b^{-1},\ldots\}$. Then the vertices of the graph can be designated by words on $\pi$, e.g. $a, a^2, a^{-1}, a.b,\ldots$, but more precisely they are the equivalence classes of these words with respect to the group equivalence $\equiv$, e.g. $b^{-1}.b.a$ and $a$ designate the same vertex. The edges are those pairs $(u,u.a)$. Cayley graphs have been used intensively because they have a number of useful features:
\begin{itemize}
\item[$\bullet$] Once an origin has been chosen, all other vertices can be named relatively to the origin.
\item[$\bullet$] The resulting graph represents the group, i.e. the set of terms and their equality.
\item[$\bullet$] There is a well-defined notion of translation of the graph, which corresponds to changing the point representing the origin, or equivalently applying an element of the group to all vertices.
\item[$\bullet$] The set of configurations over a given Cayley graph (i.e. labellings of the graph) can be given the structure of a compact metric space, which has been used in order to define Cellular Automata over them.
\end{itemize}
{ In this paper, we propose a notion of graph associated to an adjacency language $L$ and its equivalence relation $\equiv_L$, which conserves or generalizes all of these features. Whereas Cayley graphs are very regular, associated graphs are arbitrary, albeit connected and of a bounded degree.}

\noindent {\em Cellular Automata.} Cellular Automata (CA) consist of a grid of identical square cells, each of which may take one of a finite number of possible states. The entire array evolves in discrete time steps. The time evolution is required to be translation-invariant (it commutes with translations of the grid) and causal (information cannot be transmitted faster than a fixed number of cells per time step). Whilst Cellular Automata are usually defined as exactly the functions having those physics-like symmetries, it turns out that they can also be characterized in purely mathematical terms as the set of translation-invariant continuous functions \cite{Hedlund} for a certain compact metric. As a consequence CA definitions are quite naturally extended from grids to Cayley graphs, where most of the theory carries through \cite{Roka,Coornaert}. Moving on, there have been several approaches to generalize Cellular Automata not just to Cayley graphs, but to arbitrary connected graphs of bounded degree:
\begin{itemize}
\item[$\bullet$] With a fixed topology, in order to describe certain distributed algorithms \cite{PapazianRemila,Gruner,Gruner2}, or to generalize the Garden-of-Eden theorem \cite{Gromov,CeccheriniEden}.
\item[$\bullet$] Through the simulation environments of \cite{GiavittoMGS,Mammen,Kurth} which offer the possibility of applying a local rewriting rule simultaneously in different non-conflicting places. 
\item[$\bullet$] Through concrete instances advocating the concept of CA extended to time-varying graphs as in \cite{TomitaGRA,KreowskiKuske,MeyerLove}, some of which are advanced algorithmic constructions \cite{TomitaSelfReproduction,TomitaSelfDescription}.
\item[$\bullet$] Through Amalgamated Graph Transformations \cite{BFHAmalgamation,LoweAlgebraic} and Parallel Graph Transformations \cite{EhrigLowe,Taentzer,TaentzerHL}, which work out rigorous ways to apply a local rewriting rule synchronously throughout a graph.
\end{itemize}
{ The approach of this paper is different in the sense that it first generalizes Cayley graphs, and then applies the mathematical characterization of Cellular Automata as the set of translation-invariant continuous functions in order to generalize CA. Compared with the above mentioned CA papers, the contribution is to extend the fundamental structure theorems about Cellular Automata to arbitrary, connected, bounded degree, time-varying graphs. Compared with the above mentioned Graph Rewriting papers, the contribution is to deduce aspects of Amalgamated/Parallel Graph Transformations from the axiomatic and topological properties of the global function.}

\noindent {\em Causal Graph Dynamics.} The work \cite{ArrighiCGD} by Dowek and one of the authors already achieves an extension of Cellular Automata to arbitrary, bounded degree, time-varying graphs, also through a notion of continuity, with the same motivations. However, graphs in \cite{ArrighiCGD} lack a compact metric over graphs, which is left as an open question. As a consequence all the necessary facts about the topology of Cayley graphs get reproven. It also leaves open whether causal graph dynamics are computable. { These issues vanish in the new formalism; which suggests that the new formalism itself is the main contribution of this paper.} 

\noindent {\em This paper.} Section \ref{sec:AssociatedGraphs} provides a generalization of Cayley graphs. This takes the form of an isomorphism between graphs and languages endowed with an equivalence.
Section \ref{sec:Operations} provides basic operations upon generalized Cayley graphs. 
Section \ref{sec:Topology} provides facts about the topology of generalized Cayley graphs. It  follows that continuous functions are uniformly continuous.
Section \ref{sec:Causality} establishes a notion of Cellular Automata over generalized Cayley graphs. A theorem of equivalence between a mathematical and a constructive approach is given.  It also shows that recognizing valid local rules is a recursive task, as well as that of computing their effect over finite graphs; this grants our model the status of a model of computation. 
Section \ref{sec:Consequences} provides important corollaries: the stability of the notion of Cellular Automata over over generalized Cayley graphs under composability and taking the inverse. It quickly mentions the status of the Garden of Eden theorem in this setting, as well as interesting subclasses of graph dynamics. 

\section{Generalized Cayley graphs}\label{sec:AssociatedGraphs}

Basically, generalized Cayley graphs are your usual, connected, undirected, bounded-degree graphs, but with five added twists:
\begin{itemize}
\item[$\bullet$] Edges are between ports of vertices, rather than vertices themselves, so that each vertex can distinguish its different neighbours, via the port that connects to it. 
\item[$\bullet$] There is a privileged vertex playing the role of an origin, so that any vertex can be referred to relative to the origin, via a sequence of ports that lead to it. 
\item[$\bullet$] The graphs are considered modulo isomorphism, so that only the relative position of the vertices can matter.
\item[$\bullet$] The vertices and edges are given labels taken in finite sets, so that they may carry an internal state just like the cells of a Cellular Automata. 
\item[$\bullet$] The labelling functions are partial, so that we may express our partial knowledge about part of a graph. For instance is is common that a local function may yield a vertex, its internal state, its neighbours, and yet have no opinion about the internal state of those neighbours. 
\end{itemize} 
The present section is a thorough formalization and study of these Generalized Cayley graphs. A fast-track reading path is given by: the following notations, the constructive view of Definitions \ref{def:graphs}-\ref{def:pointedmodulo}, the algebraic view of Definitions \ref{def:completeness}-\ref{def:adjacencystructures}, their equivalence as discussed in Subsection \ref{subsec:equivalence}.

{\em Notations.} Let $\pi$ be a finite set, $\Pi=\pi^2$ be its square, and $V={\cal P}(\Pi^*)$ the set of languages over the alphabet $\Pi$. The operator `$.$' represents the concatenation of words and $\varepsilon$ the empty word, as usual.\\
The {\em vertices} of the graphs (see Figure \ref{fig:graphs}$(a)$) we consider in this paper are uniquely identified by a name $u$ in $V$. 
Vertices may also be labelled with a {\em state} $\sigma(u)$ in $\Sigma$ a finite set.
Each vertex has {\em ports} in the finite set $\pi$. 
A vertex and its port are written $u \port  a$.\\  
An {\em edge} is an unordered pair $\{u \port a, v \port b\}$. 
Such an edge connects vertices $u$ and $v$; we shall consider connected graphs only.
The port of a vertex can only appear in one edge, so that the degree of the graphs is always bounded by $|\ports|$. 
Edges may also be labelled with a {\em state} $\delta(\{u \port a, v \port b\})$ in $\Delta$ a finite set. 

\subsection{Graphs as paths}

\noindent Definitions \ref{def:graphs} to \ref{def:isomorphism} are as in \cite{ArrighiCGD}. The first two are reminiscent of the many papers seeking to generalize Cellular Automata to arbitrary, bounded degree, fixed graphs \cite{PapazianRemila,Gruner,Gruner2,Gromov,CeccheriniEden,TomitaGRA,KreowskiKuske,TomitaSelfReproduction,TomitaSelfDescription,BFHAmalgamation,LoweAlgebraic,EhrigLowe,Taentzer,TaentzerHL}.
They are illustrated by Figure \ref{fig:graphs}$(a)$.

\begin{definition}[Graph]\label{def:graphs}
A {\em graph} $G$ is given by 
\begin{itemize}
\item[$\bullet$] An at most countable subset $\vt(G)$ of $V$, whose elements are called {\em vertices}.
\item[$\bullet$] A finite set $\ports$, whose elements are called {\em ports}.
\item[$\bullet$] A set $E(G)$ of non-intersecting two element subsets of $\vt(G)\port\ports$, whose elements are called edges. In other an edge $e$ is of the form $\{u \port a, v \port b\}$, and $\forall e,e'\in E(G), e\cap e'\neq \emptyset \Rightarrow e=e'$. 
\end{itemize}
The graph is assumed to be connected, i.e. for any two $u,v\in V(G)$, there exists $v_1,\ldots , v_{n-1}\in V(G)$ such that for all $i\in\{0\ldots n-1\}$, one has $\{v_i\port a_i,v_{i+1}\port b_i\}\in E(G)$ with $v_0=u$ and $v_n=v$.
\end{definition}

\begin{definition}[Labelled graph]
A labelled graph is a triple $(G,\sigma,\delta)$, also denoted simply $G$ when it is unambiguous, where $G$ is a graph, and $\sigma$ and $\delta$ respectively label the vertices and the edges of $G$:
\begin{itemize}
\item[$\bullet$] $\sigma$ is a partial function from $V(G)$ to $\Sigma$;
\item[$\bullet$] $\delta$ is a partial function from $E(G)$ to $\Delta$.
\end{itemize}
The {\em set of all graphs} with ports $\ports$ is written ${\cal G}_{\ports}$.  
The {\em set of labelled graphs} with states $\Sigma,\Delta$ and ports $\ports$ is written ${\cal G}_{\Sigma,\Delta,\ports}$.
To ease notations, we sometimes write $v \in G$ for $v \in \vt(G)$.
\end{definition}
We now want to single out a vertex. The following definition is illustrated by Figure \ref{fig:graphs}$(b)$.
\begin{definition}[Pointed graph]
A {\em pointed (labelled) graph} is a pair $(G,p)$ with $p\in G$. 
The {\em set of pointed graphs}  with ports $\ports$  is written ${\cal P}_{\ports}$.
The {\em set of pointed labelled graphs} with states $\Sigma,\Delta$ and ports $\ports$ is written ${\cal P}_{\Sigma,\Delta,\ports}$.
\end{definition}

\begin{definition}[Isomorphism]\label{def:isomorphism}
An {\em isomorphism} $R$ is a function from ${\cal G}_{\pi}$ to ${\cal G}_{\pi}$ which is specified by a bijection $R(.)$ from $V$ to $V$. 
The image of a graph $G$ under the isomorphism $R$ is a graph $RG$ whose set of vertices is $R(\vt(G))$, and whose set of edges is $\{\{R(u):a,R(v):b\} \;|\; \{u:a,v:b\}\in E(G) \}$. 
Similarly, the image of a pointed graph $P=(G,p)$ is the pointed graph $RP=(RG,R(p))$. 
When $P$ and $Q$ are isomorphic we write $P\approx Q$, defining an equivalence relation on the set of pointed graphs. The definition extends to pointed labelled graphs.
\end{definition}

\noindent In the particular graphs we are considering, the vertices can be uniquely distinguished by the paths that lead to them starting from the pointer vertex. Hence, we might just as well forget about vertex names. The following definition is illustrated by Figure \ref{fig:graphs}$(c)$.
\begin{definition}[Pointed graph modulo]\label{def:pointedmodulo}
Let $P$ be a pointed (labelled) graph $(G,p)$. The {\em pointed (labelled) graph modulo} $\tilde{P}$ is the equivalence class of $P$ with respect to the equivalence relation $\approx$. The {\em set of pointed graphs modulo} with ports $\ports$ is written ${\cal \tilde{P}}_{\ports}$. The {\em set of pointed labelled graphs modulo} with states $\Sigma,\Delta$ and ports $\ports$ is written ${\cal \tilde{P}}_{\Sigma,\Delta,\ports}$.
\end{definition}
\noindent These pointed graph modulo will constitute the set of {\em configurations} (a.k.a. {\em generalized Cayley graphs}) of the generalized Cellular Automata that we will consider in this paper. For now, we want to reach a more algebraic description of them. Indeed, given such a pointed graph modulo, its set of paths forms a language, endowed with a notion of equivalence whenever two paths designate the same vertex. The language, together with its equivalence, is referred to as a path structure.

\begin{figure}

\includegraphics[scale=1]{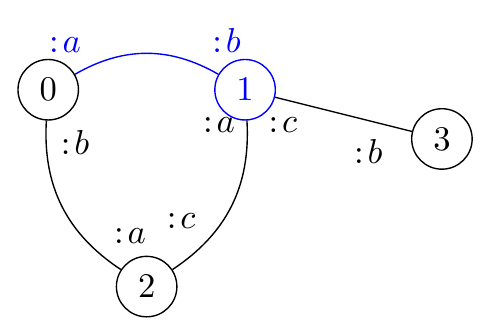}
\caption{\label{fig:graphs} {\em The different types of graphs.} (a) A graph. (b) A pointed graph. (c) A pointed graph modulo. We will see in the next sections that (c) can also be described as a language $\{\varepsilon,ab, ab.ba,\ldots\}$ and an equivalence relation with equivalence classes corresponding to vertices: $\tilde{\varepsilon}=\{\varepsilon,ab.ba,\ldots\}$ and $\tilde{ab}=\{ab,ab.ba.ab,\ldots\}$ }
\end{figure}

\begin{definition}[Path]
Given a pointed graph modulo $\tilde{P}$, we say that $\alpha$ is a path of $\tilde{P}$ if and only if there is a sequence $\alpha$ of ports $a_ib_i$ such that, starting from the pointer, it is possible to travel in the graph according to this sequence. More formally, $\alpha$ is a path if and only if there exists $(G,p)\in \tilde{P}$ and $v_1,\ldots , v_{n}\in V(G)$ such that for all $i\in\{0\ldots n-1\}$, one has $\{v_i\port a_i,v_{i+1}\port b_i\}\in E(G)$, with $v_0=p$ and $\alpha_i=a_ib_i$. Notice that the existence of a path does not depend on the choice of $P\in \tilde{P}$. The {\em language of paths} of $\tilde{P}$ is written $L(\tilde{P})$, and is set of all the paths of $\tilde{P}$.
\end{definition}

\begin{definition}[Equivalence of paths]
Given a pointed graph modulo $\tilde{P}$, we define the {\em equivalence of paths} relation $\equiv_{\tilde{P}}$ on $L(\tilde{P})$ such that for all paths $u,u'\in L(\tilde{P})$, $u\equiv_{\tilde{P}} u'$ if and only if, starting from the pointer, $u$ and $u'$ lead to the same vertex of $\tilde{P}$. 
More formally, $u\equiv_{\tilde{P}}u'$ if and only if there exists $(G,p)\in \tilde{P}$ and $v_1,\ldots , v_{|u|},v'_1,\ldots , v'_{|u'|}\in V(G)$ such that for all $i\in\{0\ldots |u|-1\}$, $i'\in\{0\ldots |u'|-1\}$, one has $\{v_i\port a_i,v_{i+1}\port b_i\}\in E(G)$, $\{v'_{i'}\port a'_{i'},v'_{i'+1}\port b'_{i'}\}\in E(G)$, with  $v_0=p$, $v'_0=p$, $u_i=a_ib_i$, $u'_{i'}=a'_{i'}b'_{i'}$ and $v_{|u|}=v_{|u'|}$.
\end{definition}

\begin{definition}[Path structure]\label{def:associatedstructure}
Given a pointed graph modulo $\tilde{P}$, we define the {\em structure of paths} $X(\tilde{P})$ as the structure $\langle L(\tilde{P}),\equiv_{\tilde{P}}\rangle$. The {\em set of all path structures} is the set $\{X(\tilde{P})\;|\;\tilde{P}\in {\cal \tilde{P}}_{\ports}\}$. It is written $X({\cal \tilde{P}}_{\ports})$.
\end{definition}

\noindent Given two pointed graphs modulo, any difference between them shows up in their path structure.
\begin{proposition}[Pointed graphs modulo and path structures isomorphism]\label{prop:graphsaspaths}
The function $\tilde{P}\mapsto X(\tilde{P})$ 
is a bijection between ${\cal \tilde{P}}_{\ports}$ and $X({\cal \tilde{P}}_{\ports})$.
\end{proposition}
\VL{
\begin{proof}
$[$Surjectivity$]$. By definition of $X({\cal \tilde{P}}_{\ports})$.\\
$[$Injectivity$]$.
Let us suppose that $X(\tilde{P})=X(\tilde{Q})$. Then $\equiv_{\tilde{P}}$ and $\equiv_{\tilde{Q}}$ must have the same number of equivalence classes and $|V(\tilde{P})|=|V(\tilde{Q})|$. Let us choose two graphs $P\in \tilde{P}$ and $Q\in \tilde{Q}$.
 For any vertex $u$ of $P$, there is a unique equivalence class $c$ of $\equiv_{\tilde{P}}$ such that the paths of $c$ lead to $u$ in $P$. Since $\equiv_{\tilde{P}}$ and $\equiv_{\tilde{Q}}$ are supposed equal, $c$ is also an equivalence class of $\equiv_{\tilde{Q}}$.
 Conversely given $c$ an equivalence class of $\equiv_{\tilde{Q}}$, there is a unique $v$ of $Q$ such that the paths of $c$ lead to $v$ in $Q$. 
Then, the paths which point to $u$ in $P$ are the same as those which point to $v$ in $Q$. We can now define a function $R$ which maps each vertex $u$ in $P$ to its corresponding vertex $v$ in $Q$. Because this is a  bijection, we can then extend $R$ to be a bijection over the entire set $V$. Let us consider two vertices $u$ and $u'$ in $P$ linked by and edge $\{u:i,u':j\}$ and their corresponding vertices $v$ and $v'$ in $Q$. As $P\in\tilde{P}$, we have that the equivalence classes $\tilde{u}.ij=\tilde{u'}$. As the classes representing $v$ and $v'$ are equal to $\tilde{u}.ij$ and $\tilde{u'}$. Thus $R$ is a graph isomorphism, and $P$ and $Q$ are isomorphic. This is true for every $P\in \tilde{P}$ and $Q\in \tilde{Q}$ thus $\tilde{P}=\tilde{Q}$.\qed
\end{proof}
}

\subsection{Paths as languages}

\noindent Inversely, we could have started by defining a certain class of languages endowed with an equivalence, namely adjacency structures, and then asked whether the path structures of graph modulo fall into this class. This is the purpose of the following definitions and lemma.

\begin{definition}[Completeness]\label{def:completeness}
Let $L\subseteq \Pi^*$ be a language and $\equiv_L$ an equivalence on this language. The tuple $(L,\equiv_L)$ is said to be {\em complete} if and only if
\begin{itemize}
\item[(i)] $ \forall u,v\in \Pi^*\quad u.v \in L \Rightarrow u\in L$
\item[(ii)] $\forall u,u'\in L\,\forall v\in\Pi^*\quad (u\equiv_L u'\,\wedge\, u.v\in L) \Rightarrow (u'.v\in L \,\wedge\, u'.v\equiv_L u.v)$
\item[(iii)] $\forall u\in L\,\forall a,b\in \ports \quad u.ab\in L \Rightarrow (u.ab.ba\in L \,\wedge\, u.ab.ba\equiv_L u)$
\end{itemize}
\end{definition}
The completeness conditions aim at making sure that $(L,\equiv_L)$, seen as some algebra of paths, is complete. Indeed: $(i)$ means that ``a shortened path remains a path''; $(ii)$ means that ``Different possible paths from $A$ to $B$, and then a path from $B$ to $C$, must lead to different possible paths from $A$ to $C$''; $(iii)$ means that ``if a step takes you from $A$ to $B$, the inverse step takes you from $B$ to $A$''. 

\begin{definition}[Adjacency structure]\label{def:adjacencystructures}
Let $L\subseteq \Pi^*$ be a language and $\equiv_L$ an equivalence on this language. The tuple $(L,\equiv_L)$ defines an {\em adjacency structure} if and only if it is complete and
$$\forall u,u'\in L\,\forall a,b,c \in \ports \quad (u\equiv_L u'\wedge u.ab\in L \wedge u'.ac\in L) \Rightarrow b=c.$$
When this is the case, $L$ is referred to as an {\em adjacency language} and $\equiv_L$ as an {\em adjacency equivalence}. 
We denote by $\langle L,\equiv \rangle$ an adjacency structure of langage $L$ and equivalence relation $\equiv$.  The {\em set of all adjacency structures} is written ${\cal X}_{\ports}$. From now on, $X$ will represent an element of ${\cal X}_{\ports}$.
\end{definition}
The added adjacency structure condition aims at making sure that $(L,\equiv_L)$, seen as some algebra of paths, is port-unambiguous, meaning that ``once at some place $A$, taking port $a$ leads to a definite place $B$''.

\begin{definition}[Associated (pointed) graph (modulo)]\label{def:associatedgraph}
Let $X$ be some adjacency structure $\langle L,\equiv_L \rangle$. Let $P(X)$ be the pointed graph $(G(X),\tilde{\varepsilon})$, with $G(X)$ such that:
\begin{itemize}
\item[$\bullet$] The set of vertices $V(G(X))$ is the set of equivalence classes of $X$;
\item[$\bullet$] The edge $\{\tilde{u}\port a,\tilde{v}\port b\}$ is in $E(G(X))$ if and only if $u.ab \in L$ and $u.ab\equiv_L v$, for all $u\in \tilde{u}$ and $v\in \tilde{v}$.
\end{itemize}
We define the {\em associated graph} to be $G(X)$.
We define the {\em associated pointed graph} to be $P(X)$.
We define the {\em associated pointed graph modulo} to be $\tilde{P}(X)$.
\end{definition}
\noindent {\em Soundness:} The properties of adjacency structures ensure that the ports of the vertices are not used several times. Moreover, $G(X)$ (and thus $P(X)$ are connected as every vertex is path connected to the vertex $\tilde{\varepsilon}$.

\begin{definition}[Labelled adjacency structure]
Let $X=\langle L,\equiv_L  \rangle$ be a an adjacency structure. 
A labelling with states $\Sigma,\Delta$ is given by a labelling for $G(X)$.
The {\em set of labelled adjacency structures} with states $\Sigma,\Delta$ and ports $\ports$ is written ${\cal X}_{\Sigma,\Delta,\ports}$.
\end{definition}
\noindent These labelled adjacency structure (a.k.a. {\em generalized Cayley graphs}) will constitute the set of {\em configurations} of the generalized Cellular Automata that we will consider in this paper. They are the algebraic counterpart of Definition \ref{def:pointedmodulo}, as we shall now prove.  

\VL{
\begin{lemma}[Path structures are adjacency structures]
Let $\tilde{P}$ be a pointed graph modulo. Then $X(\tilde{P})$ is an adjacency structure.
Hence  $X({\cal \tilde{P}}_{\ports})\subseteq {\cal X}_{\ports}$.
\end{lemma}}
\VL{
\begin{proof}
$[$Completeness$]$. If $u.v$ is a valid path in $\tilde{P}$, then the truncated path $u$ is a valid path in $\tilde{P}$ and belongs to $L(\tilde{P})$.\\
If two paths $u$ and $v$ in $\tilde{P}$ lead to the same vertex, i.e. $u \equiv_{\tilde{P}} v$, then extending $u$ and $v$ by the same path $w$ will still lead to the same vertex  i.e.  if $u.w \in L(\tilde{P})$  $u.w \equiv_{\tilde{P}} v.w$.\\
If $u.ab$ is a valid path in $\tilde{P}$ then the extension $u.ab.ba$ consisting in going back on the last visited vertex is still a valid path and leads to the vertex pointed by $u$.\\
Summarizing, the completeness properties are verified by construction of the language of path $L(\tilde{P})$ and the relation $\equiv_{\tilde{P}}$.\\
$[$Adjacency structure$]$.
Let us consider two paths $u$ and $v$ in $L(\tilde{P})$ and three ports $a,b,c$ such that $u\equiv_{\tilde{P}} v$ and $u.ab \equiv_{\tilde{P}} u.ac$. Then, for the graph $\tilde{P}$ to be well defined we have that $b=c$.\qed
\end{proof}
}

\noindent Not only do we have that path structures are adjacency structures, but it also turns out that any adjacency structure can be generated this way, i.e. it is the path structure of some pointed graph modulo.

\begin{proposition}[Adjacency structures are path structures]\label{prop:pathsaslanguages}
Let $X$ be some adjacency structure. The equality $X=X(\tilde{P}(X))$ holds.
Hence ${\cal X}_{\ports}=X({\cal \tilde{P}}_{\ports})$.
\end{proposition}
\VL{
\begin{proof}
Let $X=\langle L,\equiv_L\rangle$ and $X'=X(\tilde{P}(X))=\langle
L',\equiv_{L'}\rangle$. Next, we will write $X\subseteq X'$ if and only if $X\subseteq X'$ and $\equiv_{L}\subseteq\equiv_{L'}$, with the relations  $\equiv_{L}$, $\equiv_{L'}$ viewed as subsets of $(L\cup L')^2$.\\

$[X\subseteq X(\tilde{P}(X))]$:

Let us consider $w\in L$. By construction of $\tilde{P}(X)$, there
exists a path $w$ in $\tilde{P}(X)$. By definition of the function
$X$, we have that this path will be represented by the word $w\in L'$.
Now, let us consider two words $u$ and $v$ in $L$ such that $u\equiv
v$. By construction of $\tilde{P}(X)$, $u$ and $v$ will be two paths
of $\tilde{P}(X)$ leading to the same vertex. By definition of the
function $X$, the two words $u$ and $v$  in $L'$ will be equivalent
regarding to the relation $\equiv'$.

$[X(\tilde{P}(X))\subseteq X]$:



Let $w'\in L'$. By definition there exists a path $\omega'$ in $\tilde{P}(X) $ labeled by $w'$ from the pointed vertex to a vertex $u$. By definition 12 there exists a word in $L$ describing the path $\omega'$, hence $w'\in L$. 
Similarly we prove the inclusion $\equiv_{L'}\subseteq\equiv_{L}$.

\end{proof}}

\subsection{Graphs as languages}\label{subsec:equivalence}

\noindent {\em Generalized Cayley graphs.} Summarizing, $X(.)$ is bijective from Proposition \ref{prop:graphsaspaths} and $X\circ \tilde{P}=Id$ from Proposition \ref{prop:pathsaslanguages}, thus $\tilde{P}$ is bijective, i.e. the following theorem comes out as a corollary:
\begin{theorem}[Pointed graphs modulo and adjacency structures isomorphism]\label{th:SP}
The function $\tilde{P}\mapsto X(\tilde{P})$ 
is a bijection between ${\cal \tilde{P}}_{\ports}$ and ${\cal X}_{\ports}$, 
whose inverse is the function $X\mapsto \tilde{P}(X)$. It can be extended into a bijection between ${\cal \tilde{P}}_{\Sigma,\Delta,\ports}$ and ${\cal X}_{\Sigma,\Delta,\ports}$.
\end{theorem}
Therefore, ${\cal \tilde{P}}_{\Sigma,\Delta,\ports}$ and ${\cal X}_{\Sigma,\Delta,\ports}$ are the same set, namely the set of {\em generalized Cayley graphs}. Our generalization of CA will have its configurations in this set.

\noindent {\em Conventions.} The above theorem justifies the fact that
\begin{itemize}
\item[$\bullet$] a (labelled) pointed graph modulo $\tilde{P}(X)$ (resp. $\tilde{P}$),
\item[$\bullet$] a (labelled) adjacency structure $X$ (resp. $X(\tilde{P})$),
\item[$\bullet$] and their associated graph $G(X)$ (resp. $G(X(\tilde{P})$)
\end{itemize}
can be viewed as three presentations of the same mathematical object. Together with Definitions \ref{def:associatedstructure} and \ref{def:associatedgraph}, it also justifies the fact that the vertices of this mathematical object can be designated by
\begin{itemize}
\item[$\bullet$]  $\tilde{u}$ an equivalence class of $X$ (resp. $X(\tilde{P})$), i.e. the set of all paths leading to this vertex starting from $\tilde{\varepsilon}$,
\item[$\bullet$] or more directly by $u$ an element of an equivalence class $\tilde{u}$ of $X$ (resp. $X(\tilde{P})$), i.e. a particular path leading to this vertex starting from $\varepsilon$.
\end{itemize}
These two remarks lead to the following mathematical conventions, which we adopt for convenience. From now on:
\begin{itemize}
\item[$\bullet$] ${\cal \tilde{P}}_{\Sigma,\Delta,\ports}$ and ${\cal X}_{\Sigma,\Delta,\ports}$ will no longer be distinguished. The latter notation will be preferred. We shall speak of a ``generalized Cayley graph'' $X$ in ${\cal X}_{\Sigma,\Delta,\ports}$.
\item[$\bullet$] $\tilde{u}$ and $u$ will no longer be distinguished. The latter notation will be given the meaning of the former. I.e. we shall speak of a ``vertex'' $u$ in $V(X)$ (or simply $u\in X$.
\item[$\bullet$] It follows that `$\equiv$' and `$=$' will no longer be distinguished. The latter notation will be given the meaning of the former. I.e. we shall speak of ``equality of vertices'' $u=v$ (when strictly speaking we just have $\tilde{u}=\tilde{v}$).
\end{itemize} 
In any case, we will make sure that a rigorous meaning can always be recovered by placing tildes back.

\noindent {\em Discussion.} Generalized Cayley graphs extend Cayley graphs:
\begin{proposition}[Recovering Cayley graph]
Consider $H$ a group with law $*$ and generators the finite set $h=\{a,b,\ldots\}$. 
Let $\pi=\{a,a^{-1}\;|\;a\in h\}$ be the generators together with their inverses, $\overline{\pi}=\{(a,a^{-1}),a^{-1}a\;|\;a\in\pi\}$ the generators paired up with their inverses. Notice that $L=\overline{\pi}^*\subset\Pi^*=(\pi^2)^*$. Consider the morphism mapping:
\begin{itemize}
\item[$\bullet$] $a$ in $\pi$ to $\overline{a}=(a,a^{-1})$ in $\overline{\pi}$
\item[$\bullet$] the term $a*v$ in $H$ to $\overline{a}.\overline{v}$ in $L$
\item[$\bullet$] the equivalence $u=v$ over $H$ to the equivalence $\overline{u}\equiv_L\overline{v}$ over $L$.
\end{itemize}
Then, $X=\langle L,\equiv_L\rangle$ is an adjacency structure, and the generalized Cayley graph $X$ coincides with the Cayley graph of $H$. 
\end{proposition}
\begin{proof} All of the adjacency structures conditions are met:
\begin{itemize}
\item[(i)] $\overline{u}.\overline{v} \in L \Rightarrow \overline{u}\in L$ by definition of $L$. 
\item[(ii)] $\overline{u}\equiv_L \overline{u}' \Rightarrow \overline{u}'.\overline{v}\equiv_L \overline{u}.\overline{v}$, since $u=u' \Rightarrow u*v=u*v'$.
\item[(iii)] $\overline{u}.\overline{a} \Rightarrow \overline{u}.\overline{a}.\overline{a^{-1}}\equiv_L \overline{u}$, since $u*a*a^{-1}=u$.
\item[(-)] $(\overline{u}\equiv_L \overline{u}'\wedge \overline{u}.(a,b)\in L \wedge \overline{u}'.(a,c)\in L) \Rightarrow b=c=a^{-1}$ by definition of $L$. 
\end{itemize}
\end{proof}
One might have thought that any adjacency structure over the language $\langle L,\equiv_L\rangle$, with $L=\overline{\pi}^*$ is a Cayley graph, but this is not the case: the fact that $\equiv_L$ corresponds to group equality does matter in the above proposition. The Petersen graph, for instance, can be endowed with such an adjacency structure, while being famously not a Cayley graph \cite{Godsil}.\\
\begin{figure}\label{fig:petersen}
\begin{center}
\includegraphics[scale=1.2]{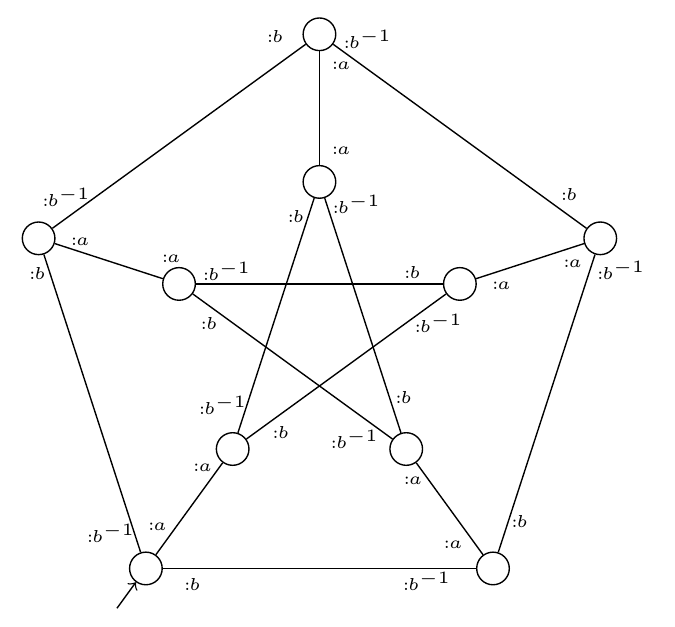}
\end{center}
\caption{{\em The Petersen graph as a generalized Cayley graph structure.}}
\end{figure}
But generalized Cayley graphs extend Cayley graphs in a much wider way than just including Petersen-like graphs. Indeed, whereas Cayley graphs are highly symmetric, generalized Cayley graphs can be {\em arbitrary connected graphs of bounded degree}. Still, this extension is an advantageous one, since all of the key features of Cayley graphs remain: We are able to name vertices relative to a point, through the word  describing the path from that point, and in fact the topology of the graph describes the equivalence structure upon words. We have a well-defined notion of translation, which is described as part of the basic operations upon these graphs in Section \ref{sec:Operations}. We can define a distance between theses graphs, which makes ${\cal X}_{\Sigma,\Delta,\pi}$ a compact metric space, as done in Section \ref{sec:Topology}.

\section{Basic operations}\label{sec:Operations}

\subsection{Operations on generalized Cayley graphs}\label{subsec:opgcg}

For a generalized Cayley graph $(G,p)$ non-modulo (see \cite{ArrighiCGD} for details):
\begin{itemize}
\item[$\bullet$] the neighbours of radius $r$ are just those vertices which can be reached in $r$ steps starting from the pointer $p$;
\item[$\bullet$] the disk of radius $r$, written $G^r_p$, is the subgraph induced by the neighbours of radius $r+1$, with labellings restricted to the neighbours of radius $r$ and the edges between them, and pointed at $p$.
\end{itemize}
Notice that the vertices of $G^r_p$ continue to have the same names as they used to have in $G$. For generalized Cayley graphs, on the other hand, the analogous operation is:
\begin{definition}[Disk]
Let $X\in {\cal X}_{\Sigma,\Delta,\ports}$ be a generalized Cayley graph and $(G,\varepsilon)$ its associated pointed graph. 
Let $X^r$ be $X(\widetilde{G^r_\varepsilon})$.
The generalized Cayley graph $X^r\in {\cal X}_{\Sigma,\Delta,\ports}$ is referred to as the {\em disk of radius $r$} of $X$. The {\em set of disks of radius $r$} with states $\Sigma,\Delta$ and ports $\ports$ is written ${\cal X}^r_{\Sigma,\Delta,\ports}$.
\end{definition}
A technical remark is that the vertices of $X^r$ no longer have quite the same names as they used to have in $X$. This is because, in a generalized Cayley graph, vertices are designated by those paths that lead to them, starting from the vertex $\varepsilon$, and there were many more such paths in $X$ than there are in its subgraph $X^r$. Still, it is clear that there is a natural inclusion $V(X^r)\incl V(X)$, meaning that $u\in X^r$ implies that there exists a unique $u'\in X$ such that $u\incl u'$. Thus, we will commonly say that a vertex of $u\in X^r$ belongs to $X$, even though technically we are referring to the corresponding vertex $u'$ of $X$. Similarly, we will commonly say that a vertex of $u'\in X$ belongs to $X^r$ when we actually mean that there is a unique vertex $u$ of $X^r$ such that $u\incl u'$.

\begin{figure}\label{fig:disk}
\includegraphics[scale=1.3]{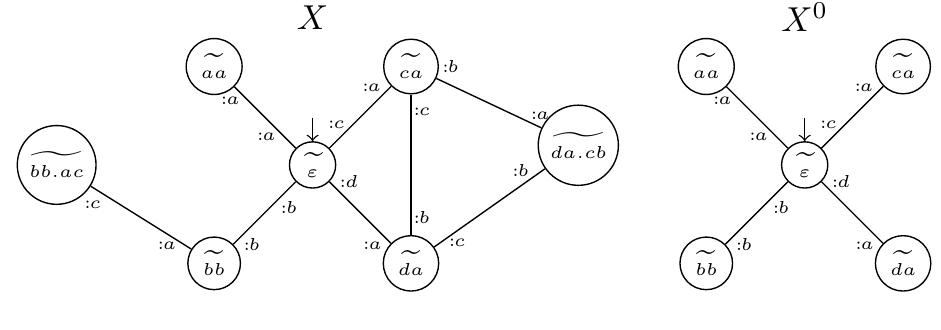}
\caption{{\em A generalized Cayley graph and its disk of radius $0$.} Notice that the equivalence classes describing vertices in $X^0$ are strict subsets of those in $X$, eventhough their shortest representative is the same. For instance the path $ca.cb$ is in $\tilde{da}$ in $X$ but is not a path in $X^0$, and thus does not belong to $\tilde{da}$ in $X^0$.}
\end{figure}

\begin{definition}[Size]\label{def:size}
Let $X\in {\cal X}_{\Sigma,\Delta,\ports}$ be a generalized Cayley graph. We say that a vertex $u\in X$ has size less or equal to $r+1$, and write $|u|\leq r+1$, if and only if $u\in X^r$.  We denote $V({\cal X}^r_\ports)=\bigcup_{X\in {\cal X}^r_\pi}V(X)$.
\end{definition}

\noindent It will help to have a notation for the graph where vertices are named relatively to some other pointer vertex $u$.
\begin{definition}[Shift]\label{def:shift}
Let $X\in {\cal X}_{\Sigma,\Delta,\ports}$ be a generalized Cayley graph and $(G,\varepsilon)$ its associated pointed graph. 
Consider $u\in X$ or $X^r$ for some $r$, and consider the pointed graph $(G,u)$, which is the same as $(G,\varepsilon)$ but with a different pointer. Let $X_u$ be $X\pa{\tili{(G,u)}}$. The generalized Cayley graph $X_u$ is referred to as {\em $X$ shifted by $u$}.\\ 
\end{definition}
\noindent Note that it could be said that $X$ and its shifted version $X_u$ are isomorphic, since both graphs are the same except for vertex naming conventions, but this is a distinct isomorphism from the isomorphism of Def. \ref{def:isomorphism}, which specifically kept the pointer unchanged.

\noindent The composition of a shift, and {\em then} a restriction, applied on $X$, will simply be written $X_u^r$. Whilst this is the analogous operation to $G^r_u$ over pointed graphs non-modulo, notice that the shift-by-$u$ completely changes the names of the vertices of $X_u^r$. As the naming has become relative to $u$, the disk $X_u^r$ holds no information about its prior location, $u$.

\noindent We may also want to designate a vertex $v$ by those paths that lead to the vertex $u$ relative to $\varepsilon$, followed by those paths that lead to $v$ relative to $u$. The following definition of concatenation coincides with the one that is induced by the concatenation of words belonging to the classes $u$ and $v$:
\begin{definition}[Concatenation]\label{def:concatenation}
Let $X\in {\cal X}_{\ports}$ be a generalized Cayley graph and $(G,\varepsilon)$ its associated pointed graph.
Consider $u\in X$ and $v\in X_u$ or $X_u^r$ for some $r$.
Let $(G',\varepsilon)$ be the associated pointed graph of $(X_u)_v$, $R$ be an isomorphism such that $G'=RG$, and $u.v$ be $R^{-1}(\varepsilon)$.
The vertex $u.v\in X$ is referred to as {\em $u$ concatenated with $v$}.
\end{definition}

\noindent According to Definition~\ref{def:shift}, $G'$ and $G$ are isomorphic.  Moreover, the restriction of $R^{-1}$ to $V(G')$ is uniquely determined; hence the definition is sound.  

\noindent It also helps to have a notation for the paths to $\varepsilon$ relative to $u$.
\begin{definition}[Inverse]
Let $X\in {\cal X}_{\ports}$ be a generalized Cayley graph and $(G,\varepsilon)$ its associated pointed graph.
Consider $u\in X$.
Let $(G',\varepsilon)$ be the associated pointed graph of $X_u$, $R$ be an isomorphism such that $G'=RG$, and $\overline{u}$ be $R(\varepsilon)$.
The vertex $\overline{u}\in X_u$ is referred to as the {\em inverse of $u$}.
\end{definition}

\noindent Notice the following easy facts: $(X_u)_v=X_{u.v}$, $u.\overline{u}=\varepsilon$. Notice also that the isomorphism $R$ such that $G(X_u)=RG(X)$ maps $v$ to $\overline{u}.v$.  
This last property suggests that we may define shifts upon graphs (non-modulo) as a certain class of isomorphisms.  In order to formalize this notion within the set of graphs without appealing to graphs modulo, we will need that the vertices of our graphs non-modulo be of a particular form.

\subsection{Operations on graphs}\label{subsec:opg}

In Section \ref{sec:AssociatedGraphs} we said that a graph $G\in {\cal G}_{\ports}$  would have vertex names in $V$. But now we shall allow vertices to have names in disjoint subsets of $V.S$, with $S=\{\varepsilon,1,2,\ldots,b\}$ a finite set of suffixes. For instance, given some generalized Cayley graph $X$, having vertices $u,v$ in $V(X)$, we may build some graph $G$ having vertices $\{v\}$, $\{u.1\}$, $\{u.3,v.1\}$ \ldots i.e. subsets of $V(X).S$. Later, $\{u.1\}$ will be interpreted as the vertex which is `the first successor of $u$', $\{u.3,v.1\}$ as the vertex which is `the first successor of $v$ and the third successor of $u$', $\{v\}$ as the vertex which is `the continuation of $v$'. Disjointness is just to keep things tidy: one cannot have a vertex which is the first successor of $u$ ($\{u.1\}$, say) coexisting with another which is the `the first successor of $u$ and the second successor of $u$' ($\{u.1,v.2\}$, say) --- although some other convention could have been used. Still, some form of suffixes is necessary in order to provide just the little, extra naming space that is needed in order to create new vertices.

\begin{definition}[Shift isomorphism]
Let $X\in {\cal X}_{\ports}$ be a generalized Cayley graph.
Let $G\in {\cal G}_{\ports}$ be a graph that has vertices that are disjoint subsets of $V(X).S$ or $V(X^r).S$ for some $r$. 
Consider $u \in X$. 
Let $R$ be the isomorphism from $V(X).S$ to $V(X_u).S$ mapping $v.z\mapsto \overline{u}.v.z$, for any $v\in V(X)$ or $V(X^r)$, $z\in S$. Extend this bijection pointwise to act over subsets of $V(X).S$, and let $\overline{u}.G$ to be $RG$.
The graph $\overline{u}.G$ has vertices that are disjoint subsets of $V(X_u).S$, it is referred to as {\em $G$ shifted by $u$}.
The definition extends to labelled graphs.
\end{definition}

\noindent The next two definitions are standard, see \cite{BFHAmalgamation,LoweAlgebraic} and \cite{ArrighiCGD}, although here again the vertices of $G$ are given names in disjoint subsets of $V(X).S$ for some $X$. Basically, we need a notion of {\em union} of graphs, and for this purpose we need a notion of {\em consistency} between the operands of the union: 
\begin{definition}[Consistency]\label{def:consistency}
Let $X\in {\cal X}_{\ports}$ be a generalized Cayley graph. Let $G$ be a labelled graph $(G,\sigma, \delta)$, and $G'$ be a labelled graph $(G',\sigma', \delta')$, each one having vertices that are pairwise disjoint subsets of $V(X).S$.  The graphs are said to be {\em consistent} if and only if:
\begin{itemize}
\item[(i)] $\forall x\in G\,\forall x'\in G'\quad x\cap x'\neq\emptyset \Rightarrow x=x'$,
\item[(ii)] $\forall x,y\in G\,\forall x',y'\in G'\,\forall a,a',b,b' \in \ports\quad (\{x\port a,y\port b\} \in E(G) \wedge \{x'\port a',y'\port b'\} \in E(G') \wedge x=x' \wedge a=a') \Rightarrow (b=b' \wedge y=y')$,
\item[(iii)] $\forall x,y\in G\,\forall x',y'\in G'\,\forall a,b \in \ports \quad x=x' \Rightarrow \delta(\{x\port a,y\port b\})=\delta'(\{x'\port a,y'\port b\})$ when both are defined,
\item[(iv)] $\forall x\in G\,\forall x'\in G'\quad  x=x' \Rightarrow \sigma(x)=\sigma'(x')$ when both are defined.
\end{itemize}
They are said to be {\em trivially consistent} if and only if for all $x\in G$, $x'\in G'$ we have $x\cap x'=\emptyset$.
\end{definition}
The consistency conditions aim at making sure that both graphs ``do not disagree''. Indeed: $(iv)$ means that ``if $G$ says that vertex $x$ has label $\sigma(x)$, $G'$ should either agree or have no label for $x$''; $(iii)$ means that ``if $G$ says that edge $e$ has label $\delta(e)$, $G'$ should either agree or have no label for $e$''; $(ii)$ means that ``if $G$ says that starting from vertex $x$ and following port $a$ leads to $y$ via port $b$, $G'$ should either agree or have no edge on port $x\port a$''.\\
Condition $(i)$ is in the same spirit: it requires that $G$ and $G'$, if they have a vertex in common, then they must fully agree on its name. Remember that vertices of $G$ and $G'$ are disjoint subsets of $V(X).S$. If one wishes to take the union of $G$ and $G'$, one has to enforce that the vertex names will still be disjoint subsets of $V(X).S$.\\ 
Trivial consistency arises when $G$ and $G'$ have no vertex in common: thus, they cannot disagree on any of the above.

\begin{definition}[Union]\label{def:union}
Let $X\in {\cal X}_{\ports}$ be a generalized Cayley graph. Let $G$ be a labelled graph $(G,\sigma, \delta)$, and $G'$ be a labelled graph $(G',\sigma', \delta')$, each one having vertices that are pairwise disjoint subsets of $V(X).S$.
Whenever they are consistent, their {\em union} is defined. The resulting graph $G\cup G'$ is the labelled graph with vertices $V(G)\cup V(G')$, edges $E(G)\cup E(G')$, labels that are the union of the labels of $G$ and $G'$.
\end{definition}

Finally, recall that for a pointed graph $(G,p)$ non-modulo $G^r_p$, is the subgraph induced by the neighbours of radius $r+1$, with labellings restricted to the neighbours of radius $r$ and the edges between them, and pointed at $p$ \cite{ArrighiCGD}.

\section{Generalized Cayley graphs: topological properties}\label{sec:Topology}

Having a well-defined notion of disks allows us to define a topology upon ${\cal X}_{\Sigma,\Delta,\ports}$, which is the natural generalization of the well-studied Cantor metric upon CA configurations \cite{Hedlund}.
\begin{definition}[Gromov-Hausdorff-Cantor metrics]\label{def:metric}
Consider the function
\begin{align*}
d:{\cal X}_{\Sigma, \Delta, \ports}\times{\cal X}_{\Sigma, \Delta, \ports} &\longrightarrow {\mathbb R}^+\\
(X,Y)&\mapsto d(X,Y)=0\quad \textrm{if }X=Y\\
(X,Y)&\mapsto d(X,Y)=1/2^r\quad \textrm{otherwise}
\end{align*}
where $r$ is the minimal radius such that $X^r \neq Y^r$.\\
The function $d(.,.)$ is such that for $\epsilon>0$ we have (with $r=\lfloor - \log_2(\epsilon)\rfloor$):
$$d(X,Y)<\epsilon \Leftrightarrow X^r = Y^r.$$
It defines an ultrametric distance.\\
\end{definition}
\VL{
\noindent {\em Soundness:}
$[$Nonnegativity, symmetry, identity of indiscernibles$]$ are
obvious.\\
$[$Equivalence$]$
\begin{align*}
d(X,Y)<\epsilon &\Leftrightarrow d(X,Y)=1/2^k \textrm{ with }k\in\mathbb{N}\wedge 1/2^k<\epsilon\\ 
&\Leftrightarrow k=\min\{r\in\mathbb{N}\;|\;X^r\neq Y^r\}\wedge 1/2^k<\epsilon\\
&\Leftrightarrow_{r=k-1}\,X^r=Y^r\textrm{ with }r\in\mathbb{N}\wedge 1/2^{r+1}<\epsilon\\
&\Leftrightarrow X^r=Y^r \textrm{ with }r=\lfloor-\log_2(\epsilon)\rfloor.
\end{align*}
$[$Ultrametricity$]$ Consider $k$ such
that  $1/2^k=d(X,Z)$ and $l$ such that  $1/2^l= d(X,Y)$.
By definition of the metric $X,Z$ differ only after index $k$ 
and $X,Y$ differ only after index $l$. Suppose
$k\leq l$ so that $Y,Z$ differ only after index $k$. But then
$d(Y,Z)=1/2^k$ which is $d(X,Z)$.\\
$[$Triangle inequality$]$ is obvious from the ultrametricity.\smallskip\\
}
\noindent The fact that generalized Cayley graphs are pointed graphs modulo, i.e. the fact that they have no ``vertex name degree of freedom'' is key to proving the following property. Indeed, compactness crucially relies on the set being ``finite-branching'', meaning that the set of possible generalized Cayley graphs, as one progressively enlarges the radius of a disk, remains finite. This does not hold for usual graphs.
\begin{lemma}[Compactness] \label{lem:compactness}
$({\cal X}_{\Sigma, \Delta, \ports},d)$ is a compact metric space, i.e. every sequence admits a converging subsequence.
\end{lemma}
\VL{\begin{proof} This is essentially K\"onig's Lemma.  Let us consider an infinite sequence of graphs $(X(n))_{n\in \mathbb{N}}$. Because $\Sigma$ and $\Delta$ are finite, and there is an infinity of elements of $(X(n))$, there must exist a graph of radius zero $X^0$ such that there is an infinity of elements of $(X(n))$ fulfilling $X(n)^0=X^0$. Choose one of them to be $X(n_0)$, i.e. $X(n_0)^0=X^0$. Now iterate: because the degree of the graph is bounded by $\pi$, and because $\Sigma$ and $\Delta$ are finite but there is an infinity of elements of $(X(n))$ having the above property, there must exist a pointed graph of radius one $X^1$ such that $(X^1)^0=X^0$ and such that there is an infinity of elements of $(X(n))$ having $X(n)^1=X^1$.  Choose one of them as $X(n_1)$, i.e. $X(n_1)^1=X^1$.  Etc.  The limit is the unique graph $X'$ having disks $X'^k=X^k$ for all $k$.\qed
\end{proof}}

\noindent Recall the difference in quantifiers between the continuity of a function $F$ over a metric space $({\cal X},d)$:
$$\forall X\in{\cal X}\,\forall \epsilon >0 \,\exists \eta >0 \,\forall Y\in{\cal X}, \quad d(X,Y)<\eta\Rightarrow d(F(X),F(Y))<\epsilon,$$ 
and its uniform continuity:
$$\forall \epsilon >0 \,\exists \eta >0 \,\forall X,Y\in{\cal X}, \quad d(X,Y)<\eta\Rightarrow d(F(X),F(Y))<\epsilon.$$ 
Uniform continuity is the physically relevant notion, as it captures the fact that $F$ does not propagate information too fast. In a compact setting, it is equivalent to simple continuity, which is easier to check and is the mathematically standard notion. This is the content of Heine's Theorem, a well-known result in general topology \cite{Fedorchuk}: given two ${\cal X}$ and ${\cal Y}$ be metric spaces and $F:{\cal X}\longrightarrow{\cal Y}$ continuous, if ${\cal X}$ is compact, then $F$ is uniformly continuous.

The implications of these topological notions for Cellular Automata were first studied in \cite{Hedlund}, with self-contained elementary proofs available in \cite{KariNotes}. For Cellular Automata over Cayley graphs a complete reference is \cite{Coornaert}. For Causal Graph Dynamics \cite{ArrighiCGD}, these implications had to be reproven by hand, due to the lack of a clear topology in the set of graphs that was considered. Here we are able rely on the topology of generalized Cayley graphs and reuse Heine's Theorem out-of-the-box, which makes the setting of generalized Cayley graphs a very attractive one in order to generalize CA.

\section{Causality and Localizability}\label{sec:Causality}
 
\noindent {\em Causality.} The notion of causality we will propose extends the known mathematical definition of Cellular Automata over grids and Cayley graphs. The extension will be a strict one for two reasons: not only the graphs become arbitrary, but they can also vary in time. 

The main difficulty we encountered when elaborating an axiomatic definition of causality from ${\cal X}_{\Sigma,\Delta,\ports}$ to ${\cal X}_{\Sigma,\Delta,\ports}$, was the need to establish a correspondence between the vertices of a generalized Cayley graph $X$, and those of its image $F(X)$. Indeed, on the one hand it is important to know that a given $u\in X$ has become $u'\in F(X)$, e.g. in order to express shift-invariance $F(X_u)=F(X)_{u'}$. But on the other hand since $u'$ is named relative to $\varepsilon$, its determination requires a global knowledge of $X$.

The following analogy provides a useful way of tackling this issue. Say that we were able to place a white stone on the vertex $u\in X$ that we wish to follow across evolution $F$. Later, by observing that the white stone is found at $u'\in F(X)$, we would be able to conclude that $u$ has become $u'$. This way of grasping the correspondence between an image vertex and its antecedent vertex is a local, operational notion of an observer moving across the dynamics. 

\begin{definition}[Dynamics]\label{def:dynamicsmodulo}
A dynamics $(F,R_{\bullet})$ is given by
\begin{itemize}
\item[$\bullet$] a function $F:{\cal X}_{\Sigma,\Delta,\ports}\to{\cal X}_{\Sigma,\Delta,\ports}$;
\item[$\bullet$] a map $R_{\bullet}$, with $R_{\bullet}: X\mapsto R_X$ and $R_X: V(X) \to V(F(X))$.
\end{itemize}
For all $X$, the function $R_X$ can be pointwise extended to sets, i.e. $R_X:{\cal P}(V(X)) \to {\cal P}(V(F(X)))$ maps $S$ to $R_X(S)=\{R_X(u)\;|\;u\in S\}$.
\end{definition}
The intuition is that $R_X$ indicates which vertices $\{u',v',\ldots\}=R_X(\{u,v,\ldots\})\subseteq V(F(X))$ will end up being marked as a consequence of $\{u,v,...\in X\}\subseteq V(X)$ being marked. Now, clearly, the set $\{(X,{\cal P}(V(X)))\;|\;X\in {\cal X}_{\Sigma,\Delta,\ports}\}$ is isomorphic to ${\cal X}_{\Sigma',\Delta,\ports}$ with $\Sigma'=\Sigma\times\{0,1\}$. Hence, we can define the function $F'$ that maps $(X,S)\cong X'\in {\cal X}_{\Sigma',\Delta,\ports}$ to $(F(X),R_X(S))\cong F'(X')\in {\cal X}_{\Sigma',\Delta,\ports}$, and think of a dynamics as just this function $F':{\cal X}_{\Sigma',\Delta,\ports}\to {\cal X}_{\Sigma',\Delta,\ports} $. This alternative formalism will turn out to be very useful.
\begin{definition}[Shift-invariance]
A dynamics $(F,R_{\bullet})$ is said to be {\em shift-invariant} if and only if for every $X$ and $u\in X$, $v\in X_u$, 
\begin{itemize}
\item[$\bullet$] $F(X_u)=F(X)_{R_X(u)}$
\item[$\bullet$] $R_X(u.v)=R_X(u).R_{X_u}(v)$.
\end{itemize}
\end{definition}
The second condition expresses the shift-invariance of $R_{\bullet}$. Notice that  $R_X(\varepsilon)=R_X(\varepsilon).R_X(\varepsilon)$; hence $R_X(\varepsilon)=\varepsilon$.\\
In the $F':{\cal X}_{\Sigma',\Delta,\ports}\to{\cal X}_{\Sigma',\Delta,\ports}$ formalism, the two above conditions are equivalent to just one: $F'(X_u)=F'(X)_{R_X(u)}$.  

\begin{definition}[Continuity]\label{def:continuitymodulo}
A dynamics $(F,R_{\bullet})$ is said to be {\em continuous} if and only if:
\begin{itemize}
\item[$\bullet$] $F:{\cal X}_{\Sigma,\Delta,\ports}\to{\cal X}_{\Sigma,\Delta,\ports}$ is continuous,
\item[$\bullet$] For all $X$, for all $m$, there exists $n$ such that for all $X'$, $X'^n=X^n$ implies $\dom\,R_{X'}^m\subseteq V(X'^n)$, $\dom\,R_{X}^m\subseteq  V(X^n)$ and $R_{X'}^m=R_{X}^m$.
\end{itemize}
where $R_{X}^m$ denotes the partial map obtained as the restriction of $R_X$ to the codomain $F(X)^m$, using the natural inclusion of $F(X)^m$ into $F(X)$.
\end{definition}
The second condition expresses the continuity of $R_{\bullet}$. It can be reinforced into uniform continuity: for all $m$, there exists $n$ such that for all $X$, $X'$, $X'^n=X^n$ implies $R_{X'}^m=R_{X}^m$.\\
Indeed, in the $F':{\cal X}_{\Sigma',\Delta,\ports}\to{\cal X}_{\Sigma',\Delta,\ports}$ formalism, the two above conditions are equivalent to just one: $F'$ continuous. But since continuity implies uniform continuity upon the compact space ${\cal X}_{\Sigma',\Delta,\ports}$, it follows that $F'$ is uniformly continuous, and thus the reinforced second condition.\\
We need one third, last condition:
\begin{definition}[Boundedness]\label{def:boundednessmodulo}
A dynamics $(F,R_{\bullet})$ from ${\cal X}_{\Sigma,\Delta,\ports}$ to ${\cal X}_{\Sigma,\Delta,\ports}$ is said to be {\em bounded} if and only if there exists a bound $b$ such that for all $X$, for all $w'\in F(X)$, there exist $u'\in \im R_X$ and $v'\in F(X)_{u'}^b$ 
such that $w'=u'.v'$.
\end{definition}

\noindent The following is our main definition:
\begin{definition}[Causal dynamics]\label{def:causal}
A dynamics is {\em causal} if it is shift-invariant, continuous and bounded.
\end{definition}

An example of causal dynamics is the inflating grid dynamics illustrated in Figure \ref{fig:inflatingglobal}. In the inflating grid dynamics each vertex gives birth to four distinct vertices, such that the structure of the initial graph is preserved, but inflated. The graph has maximal degree $4$, and the set of ports is $\pi=\{a,b,c,d\}$, vertices and edges are unlabelled.
\begin{figure}
\begin{tabular}{ccc}
\raisebox{1.45cm}{\includegraphics[scale=0.2]{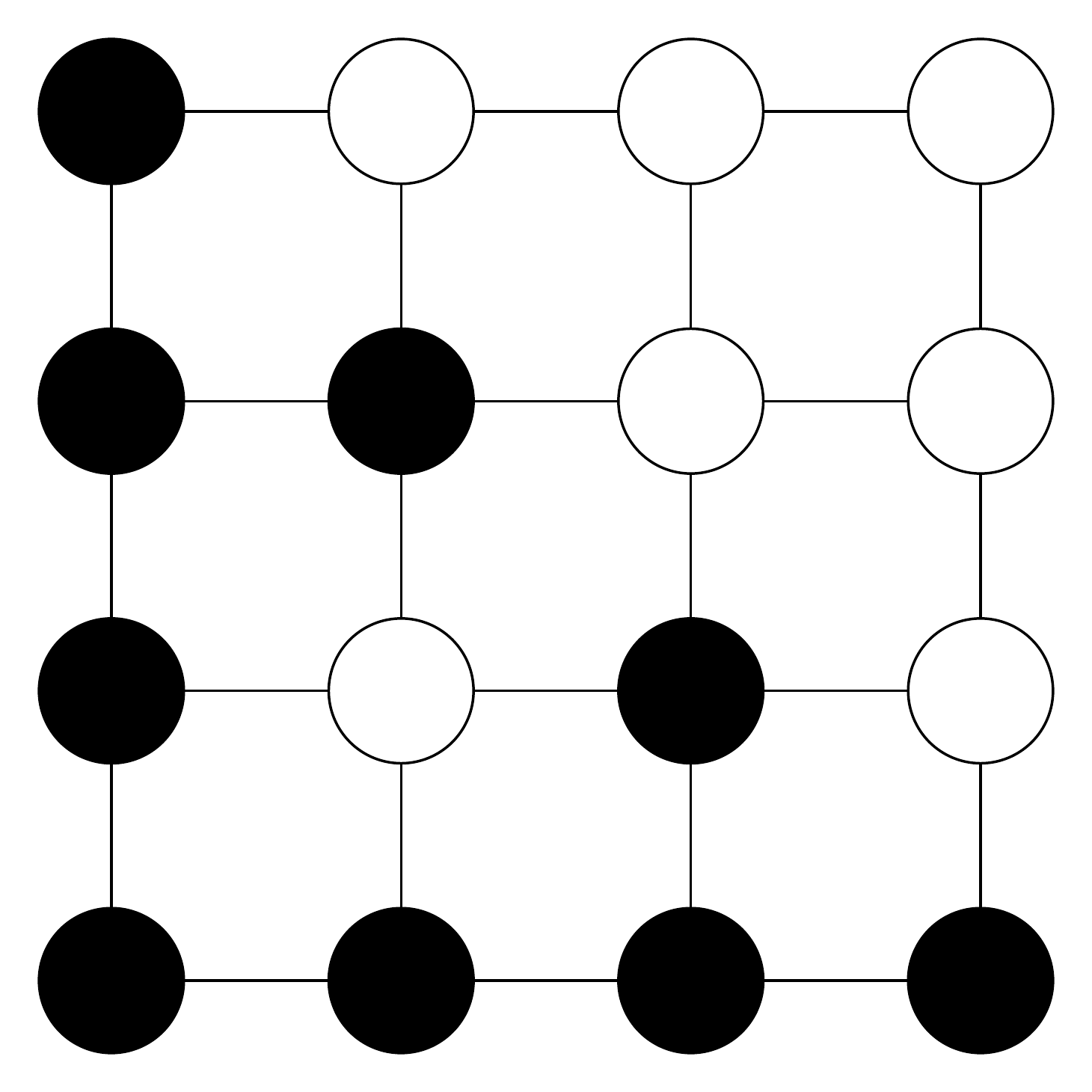}} ~~~~~~~&
\raisebox{2.75cm}{$\mapsto$}  ~~~~~~~&
\includegraphics[scale=0.4]{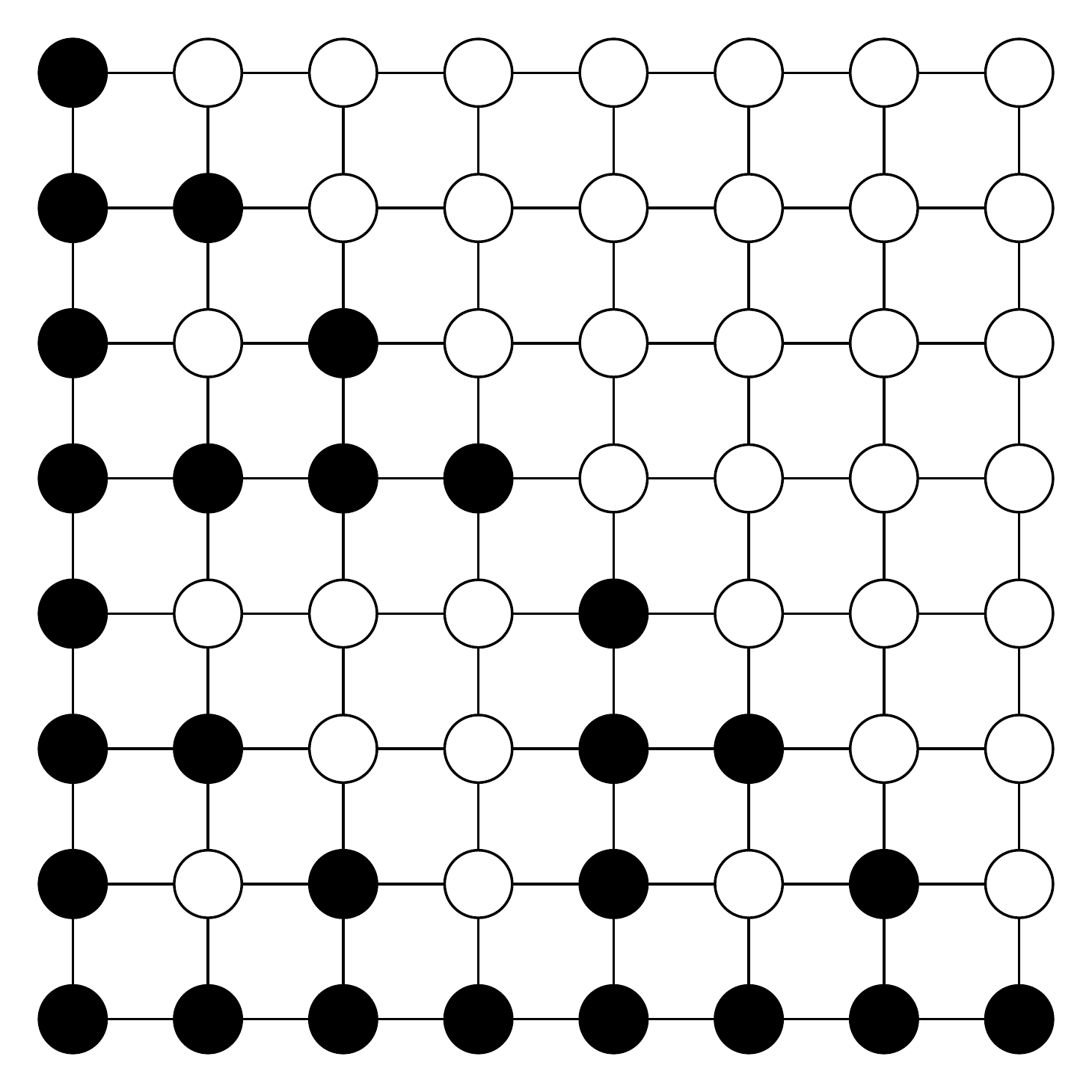}
\end{tabular}

\caption{\label{fig:inflatingglobal} {\em The inflating grid dynamics.} Each vertex splits into 4 vertices. The structure of the grid is preserved. Ports are omitted here.}
\end{figure}

\begin{lemma}[Bounded inflation]
\label{lem:connectivitymodulo}
Consider a causal dynamics $F$ from ${\cal X}_{\Sigma, \Delta, \ports}$ to ${\cal X}_{\Sigma, \Delta, \ports}$. There exists a bound $b$ such that for all $X$ and $u\in X^r$, we have $|R_{X}(u)|\leq (r+1)b$.
\end{lemma}
\begin{proof}
Let $ac\in\Pi$, and let $E$ the subset of ${\cal X}_{\Sigma, \Delta, \ports}$ of those $X$ such that $ac\in X$. $E$ is closed --- any sequence of elements of $E$ converging in  ${\cal X}_{\Sigma, \Delta, \ports}$ converges in $E$ --- and ${\cal X}_{\Sigma, \Delta, \ports}$ is compact, therefore $E$ is compact.  By continuity modulo, the function $X\mapsto |R_X(ac)|$ is continuous from $E$ to $\N$; since $E$ is compact, it must be bounded.  The result then follows from the triangle inequality and shift-invariance.
\end{proof}

\noindent {\em Localizability.} The notion of localizability of a function $F$ captures the exact same idea as the constructive definition of a Cellular Automata, namely that $F$ arises as a single local rule $f$ applied synchronously and homogeneously across the input graph. \\
The general idea is that the local rule $f$ looks at part of the generalized Cayley graph $X$ (a disk $X^r$) and produces a piece of graph $G=f(X^r)$. The same is done synchronously at every location $u\in X$ producing pieces of graph $G'=f(X_u^r)$. The produced pieces must be consistent (see Subsection \ref{subsec:opg}) so that we take their union. Their union is a graph, but taking its modulo leads to a generalized Cayley graph $F(X)$.\\ We now formalize this idea. First, we must make sure that a local rule is an object that adopts the same naming conventions for vertices as those of the basic graph operations of Subsection \ref{subsec:opg}.
\begin{definition}[Dynamics non-modulo]\label{def:dynamics}
A function $f$ from ${\cal X}^r_{\Sigma,\Delta,\ports}$ to ${\cal G}_{\Sigma,\Delta,\ports}$ is said to be a {\em dynamics} if and only if for all $X$ the vertices of $f(X)$ are disjoint subsets of $V(X).S$, and $\varepsilon\in f(X)$.
\end{definition}
Intuitively, the integer $z\in S$ stands for the `successor number $z$'. Hence the vertices designated by $\{1\},\{2\}\ldots$ are successors of the vertex $\varepsilon$, whereas $\{\varepsilon\}$ is its `continuation'. The vertices designated by $\{ab.1\},\{ab.2\}\ldots$ are successors of its neighbour $ab\in X^r$. A vertex named $\{1,ab.3\}$ is understood to be both the first successor of vertex $\varepsilon$ and the third successor of vertex $ab$. Recall also that $\varepsilon$, just like $ab$, are not just words but entire equivalence classes of these words, i.e. elements of $V(X)$.\\
Second, we disallow local rules that would suddenly produce an infinite graph.
\begin{definition}[Boundedness non-modulo]
A function $f$ from ${\cal X}^r_{\Sigma,\Delta,\ports}$ to ${\cal G}_{\Sigma,\Delta,\ports}$ is said to be {\em bounded} if and only if for all $X$, the graph $f(X)$ is finite.
\end{definition}
Third, we make sure that the pieces of graphs that are produced by the local rule are consistent with one another.
\begin{definition}[Local rule]
A function $f$ from ${\cal X}^r_{\Sigma, \Delta, \ports}$ to ${\cal G}_{\Sigma, \Delta, \ports}$ is a {\em local rule} if and only if it is a bounded dynamics and
\begin{itemize}
\item[$\bullet$] For any disk $X^{r+1}$ and any $u\in X^{0}$ we have that $f(X^r)$ and $u.f(X_u^r)$ are non-trivially consistent.
\item[$\bullet$] For any disk $X^{3r+2}$ and any $u\in X^{2r+1}$ we have that $f(X^r)$ and $u.f(X_u^r)$ are consistent.
\end{itemize}
\end{definition}
It is clear that we do not need to formulate any consistency condition beyond $u\in X^{2r+1}$, because $f(X^r)$ and $u.f(X_u^r)$ then become trivially consistent, as they share nothing in common, see Figure \ref{fig:consistency}. The only subtlety in the above definition is to impose that within $u\in X^{0}$, the produced pieces of graphs $f(X^r)$ and $u.f(X_u^r)$ be non-trivially consistent, i.e. consistent and overlapping, see Figure \ref{fig:consistency}. The point here is to enforce the connectedness of the union of the pieces of graphs via a local, syntactic restriction. 
\begin{figure}
\includegraphics[scale=1]{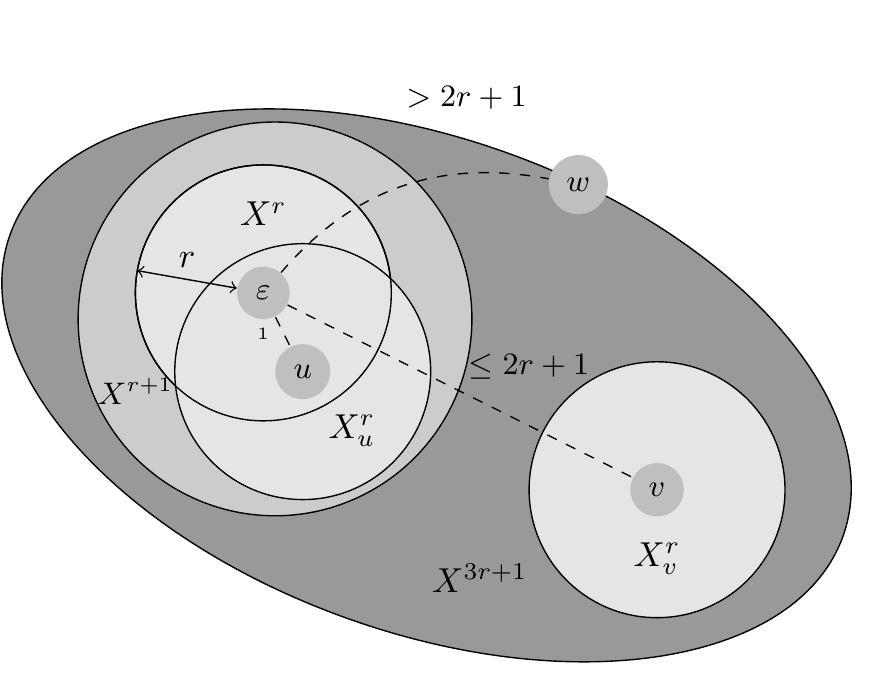}
\caption{\label{fig:consistency}  {\em The consistency conditions for a local rule.} The drawing represents disks of a generalized Cayley graph $X$ upon which a local rule $f$ of radius $r$ will be applied. $f(X^r)$ and $u.f(X_u^r)$ have to be non-trivially consistent since $\varepsilon$ and $u$ are at distance $1$. $f(X^r)$ and $v.f(X^r_v)$ have to be consistent but their intersection is allowed to be empty. $f(X^r)$ and $w.f(X^r_w)$ will be trivially consistent as they are to far to interact in one time step. The disk $X^{r+1}$ is enough to check all the non-trivial consistency conditions, as it comprises first neighbours and their $r$-disks. The disk $X^{3r+1}$ is enough to check all the consistency conditions, as it comprises all the $2r+1$ neighbours and their $r$-disks.}
\end{figure}
To illustrate the concept of local rule, we will now describe a local rule implementing the inflating grid dynamics (see Figure \ref{fig:inflatingglobal}).
The local rule is of radius zero: it ``sees'' the neighbouring vertices and nothing more.
In the standard case the local rule is applied on a vertex surrounded by $4$ neighbours. It then generates a graph of $12$ vertices, each with particular names (see Figure \ref{fig:inflatinggeneral}). In particular cases, when less than $4$ neighbours are present, the rule generates a graph of $10$, $8$, $6$ or $4$ vertices, each with particular names (see Figure \ref{fig:inflatingsecond}). The local rule is not exhaustively described here, since there exists $625$ different neighbourhoods of radius $0$. In any case, all generated vertex names are carefully chosen, so that when taking the union of all the generated subgraphs, the name collisions lead to the desired identification of vertices (see Figure \ref{fig:inflatingglueing}).
\begin{figure}
\includegraphics[scale=1.2]{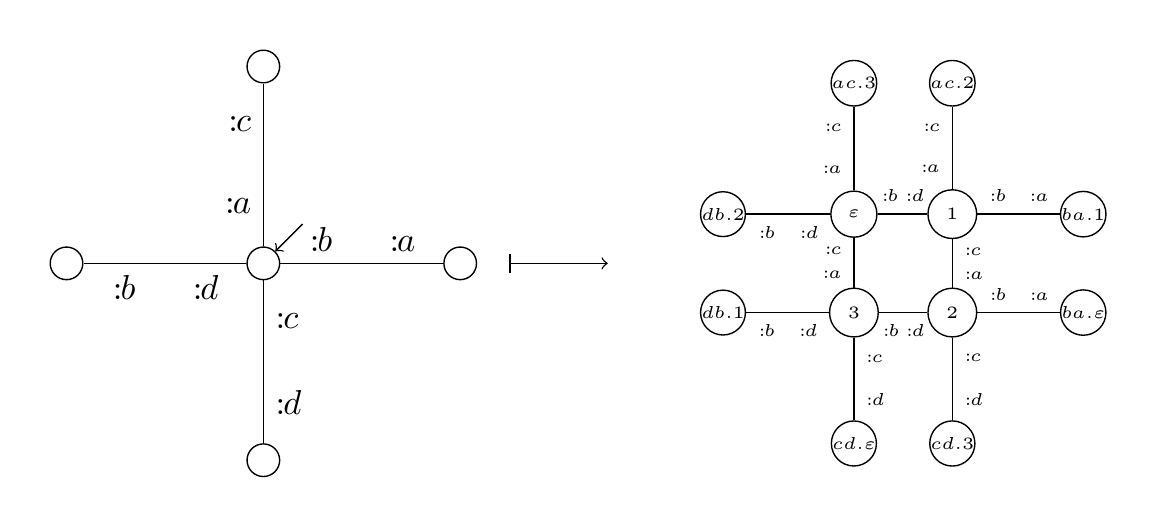}
\caption{\label{fig:inflatinggeneral} {\em Standard case of the inflating grid local rule.} The left-hand-side of the rule is a generalized Cayley graph of form $X^0_u$ (a disk of radius $0$). The right-hand-side is a graph whose vertex names are subsets of $V(X^0_u).S$. Here they are just singletons, curly brackets are dropped: e.g. we wrote $ac.3$ for $\{ac.3\}$, which should be understood as ``the third successor of my neighbour on edge $ac$''.}
\end{figure}
\begin{figure}
\includegraphics[scale=1.2]{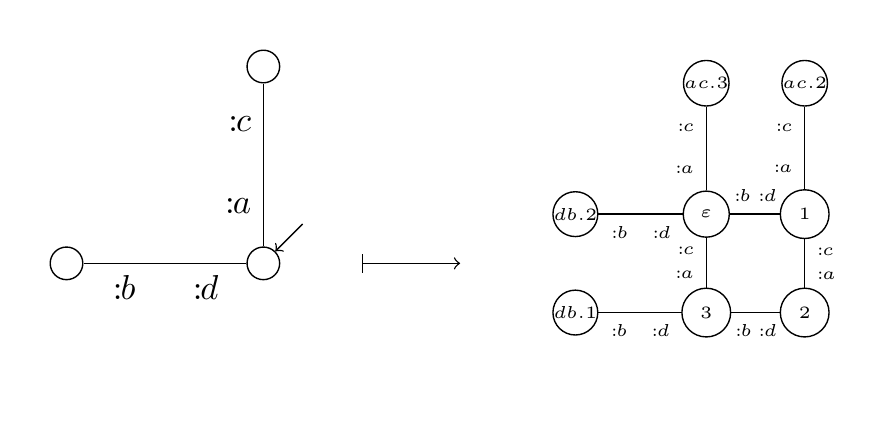}
\caption{\label{fig:inflatingsecond} {\em A particular case of the inflating grid local rule.} }
\end{figure}
\begin{figure}
\includegraphics[scale=1.2]{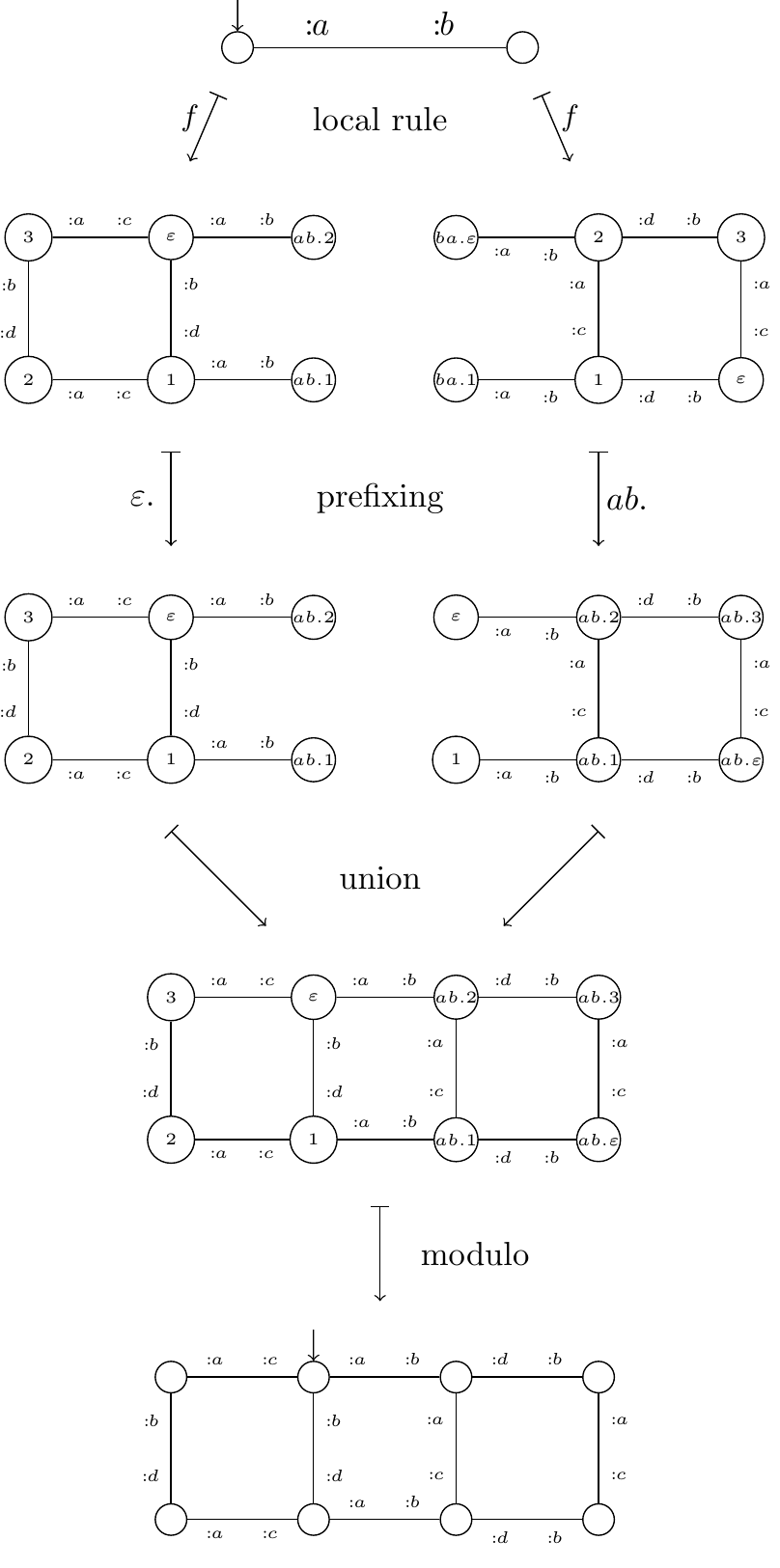}
\caption{\label{fig:inflatingglueing} {\em Local rule implementation of the inflating grid dynamics.} First the local rule is applied on the neighbourhood of every vertex of the input graph. The obtained graphs are prefixed (see definition \ref{def:concatenation}) by the vertex they are issued of. Third a union of graphs is performed to obtain the output graph. Lastly, the corresponding pointed graph modulo is returned.}
\end{figure}

\begin{definition}[Localizable function]\label{def:localizable}
A function $F$ from ${\cal X}_{\Sigma, \Delta, \ports}$ to ${\cal X}_{\Sigma, \Delta, \ports}$ is said to be {\em localizable} if and only if there exists a radius $r$ and a local rule $f$ from ${\cal X}^r_{\Sigma, \Delta, \ports}$ to ${\cal G}_{\Sigma, \Delta, \ports}$ such that for all $X$, $F(X)$ is given by the equivalence class modulo isomorphism, of the pointed graph
$$\bigcup_{u\in X} u.f(X_u^r)$$
with $\varepsilon$ taken as the pointer.
\end{definition}

~\medskip

\noindent {\em Equivalence.} The following theorem shows that the constructive definition (Localizable functions) is in fact equivalent to the mathematical, axiomatic definition (causal dynamics).

\begin{theorem}[Causal is localizable]\label{th:main}
Let $F$ be a function from ${\cal X}_{\Sigma, \Delta, \ports}$ to ${\cal X}_{\Sigma, \Delta, \ports}$. The function $F$ is localizable if and only if there exists $R_{\bullet}$ such that $(F,R_\bullet)$ is a causal dynamics.
\end{theorem}
\VL{
\begin{proof}
{\bf $[$Loc.$\Rightarrow$Caus.$]$} Let $F:{\cal X}_{\Sigma, \Delta, \ports}\rightarrow {\cal X}_{\Sigma, \Delta, \ports}$ be a localizable dynamics with local rule $f$ from ${\cal X}^r_{\Sigma,\Delta,\ports}$ to ${\cal G}_{\Sigma, \Delta, \ports}$: $F(X)$ is the equivalence class, with $\varepsilon$ taken as the pointer vertex, of the graph $H(X)=\bigcup u.f(X_u^r)$.\\
$[$Dynamics$]$ Using the dynamicity of the local rule $f$, for all $X^r$ we have $\varepsilon\in f(X^r)$. Therefore for all $u\in X$, we have $u\in u.f(X^r)$ and thus $u\in H(X)$. Let $R$ be an isomorphism such that $G(F(X))=RH(X)$. Let $u\in V(X)$, we define $R_X(u)$ to be $R(u')$, where $u'$ is the vertex of $H(X)$ that contains $u$ in its name. Notice that $\tili{(H(X),u)}=\tili{(R_XH(X),R_X(u))}=\tili{(G(F(X)),R_X(u))}=F(X)_{R_X(u)}$.\\
$[$Translation-invariance$]$ Take $u\in X$. We have $H(X_u)=\bigcup v.f(X_{u.v}^r)$. This is equal to $H(X_u)=\overline{u}.\bigcup u.v.f(X_{u.v}^r)$, which in turn is equal to $\overline{u}.H(X)$. Next, we have that $F(X_u)=\tili{(H(X_u),\varepsilon)}=\tili{(\overline{u}.H(X),\overline{u}.u)}=\tili{(H(X),u)}=F(X)_{R_X(u)}$. It follows that $F(X_u)=F(X)_{R_X(u)}$, and so $G(F(X_u))=\overline{R_X(u)}.G(F(X))$.  We have therefore
$$G(F(X))=R_X(u).G(F(X_u))=R_X(u).R_{X_u}H(X_u)=R_X(u).R_{X_u}\overline{u}.H(X).$$
But since the relation $G(F(X))=R H(X)$ defines $R_X$, we have proven that for all $u\in X$, $R_X=(R_X(u).R_{X_u}\overline{u}.)$. It follows that, for all $u.v\in X$, $R_X(u.v)=R_X(u).R_{X_u}(v)$.\\
$[$Boundedness$]$ for all $X$, for all $w'\in F(X)$, consider $w\in H(X)$ such that $w'=R(w)$ when $G(F(X))=RH(X)$, and $u\in X$ such that $w\in u.f(X_u^r)$. Since $\varepsilon\in f(X_u^r)$, we have $u\in u.f(X_u^r)$. Since $f$ is bounded, $w$ lies at most at a some distance $b$ of $u$ in $H(X)$. Since $G(F(X)=RH(X)$, $w'$ lies at most at a some distance $b$ of $u'=R(u)=R_X(u)$ in $F(X)$.\\
$[$Continuity$]$ The following is illustrated in Figure \ref{fig:continuity_proof}. Let $m\in\N$.  We must show that there exists $n$ such that $F(X)^m=\tilde{H}(X)^m_\varepsilon$ is determined by $X^n$.\\
Consider a sequence $v_0=\varepsilon, v_1, \ldots v_{m+1}$ of vertices of $H(X)$ such that for all $i\in \{0\ldots m\}$ there exists $e_i=(v_i:a_{i},v_{i+1}:b_{i+1})$ in $E(H(X))$.  For such an $e_i$ to exist, and given Definitions \ref{def:consistency} and \ref{def:union}, it must appear in some $u_i.f(X_{u_i}^r)$. Moreover if $\delta(e_i)$ is defined, it must be defined in some $u_i.f(X_{u_i}^r)$. Consider $u_0, u_1, \ldots u_m$ a sequence of vertices of $X$ such that this is the case. Also, since $v_{i+1}$ is a subset of $V(X).S$, there exists $w_i\in X, z_i\in S$ such that $w_i.z_i\in v_i$. Again consider $w_0=\varepsilon, w_1, \ldots w_{m+1}$ a sequence of vertices of $X$ such that this is the case.\\
\begin{figure}[htbp]
\centering
\begin{tikzpicture}[auto=center]
\tikzstyle{v}=[circle,fill=black!20]
\tikzstyle{u}=[circle,fill=black!20]
\tikzstyle{w}=[circle,fill=black!20]
\tikzstyle{xnode}=[circle,fill=black!20]
\tikzstyle{surround} = [circle,fill=blue!10,thick,draw=black,rounded corners=2mm]

\node[v] (v0) at (0,7) {$v_0=\varepsilon$};
\node[v] (v1) at (2,7) {$v_1$};
\node[v] (v2) at (4,7) {$v_2$};
\node[v] (vm) at (8,7) {$v_m$};
\node[v] (vmp) at (10,7) {$v_{m+1}$};

\node[xnode] (x0) at (0,5) {$x_0$};
\node[xnode] (x1) at (2,5) {$x_1$};
\node[xnode] (x2) at (4,5) {$x_2$};
\node[xnode] (xm) at (8,5) {$x_m$};

\node[w] (w0) at (0,3) {$w_0=\varepsilon$};
\node[w] (w1) at (2,3) {$w_1$};
\node[w] (w2) at (4,3) {$w_2$};
\node[w] (wm) at (8,3) {$w_m$};
\node[w] (wmp) at (10,3) {$w_{m+1}$};

\node[u] (u0) at (1,2) {$u_0$};
\node[u] (u1) at (3,2) {$u_1$};
\node (u2) at (5,2) {};
\node (umm) at (7,2) {};
\node[u] (um) at (9,2) {$u_m$};

\draw (v0)--(v1)--(v2) (vm)--(vmp);

\path[dashed]
	(v2) edge node {} (vm)
;

\path[<->,dotted] (v0) edge [bend left] node[above] {$\leq m+1$} (vmp) ;

\path[dashed]
	(w0) edge node[below left] {$\leq r+1$} (u0)
	(u0) edge node[below right] {} (w1)
	(w1) edge node[above right] {} (u1)
	(u1) edge node[below right] {} (w2)
	(w2) edge node {} (u2)
	(umm) edge node {} (wm)
	(wm) edge node[below left] {} (um)
	(um) edge node[below right] {$\leq r+1$} (wmp)
	(w0) edge node[right] {$\leq r+1$} (x0)
	(w1) edge node[] {} (x1)
	(w2) edge node[] {} (x2)
	(wm) edge node[] {} (xm)
;

\path[<->,dotted] (w1) edge node[above] {$\leq 2(r+1)$} (w2)
(w0) edge [bend left] node[below] {$\leq 2(m+1)(r+1)$} (wmp);


\end{tikzpicture}

\caption{{\em Proof of continuity.}\label{fig:continuity_proof}}
\end{figure}
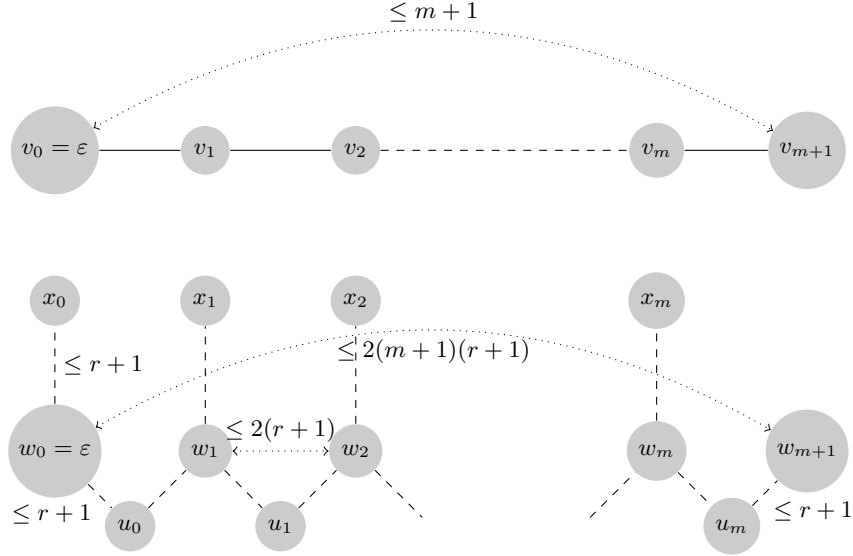
Since $e_i$ is in $u_i.f(X_{u_i})$, it follows that $v_i$ and $v_{i+1}$ are in $u_i.f(X_{u_i})$. This entails that $v_i$ and $v_{i+1}$ are subsets of $u_i.V(X_{u_i}).S$, thus in particular $w_i,w_{i+1}\in u_i.V(X_{u_i})$. 
Therefore we have both $w_{i+1}\in u_i.X_{u_i}$ and $w_{i+1}\in u_{i+1}.X_{u_{i+1}}$. As a consequence $u_i$ and $u_{i+1}$ lie at distance $2(r+1)$ in $X$, and it follows that $\bigcup_{i=0\ldots m} u_i.X_{u_i}^r \subseteq X^{2(m+1)(r+1)-1}$. 
Hence $X^{2(m+1)(r+1)-1}$ determines $E(H(X)_\varepsilon^m)$ and their internal states.\\
For $\sigma(v_i)$ to be defined, there must exists $x_i\in X$ such that $\sigma(v_i)$ is defined in $x_i.f(X_{x_i}^r)$. Consider $x_0, x_1, \ldots x_m$ a sequence of vertices of $X$ such that this is the case. But since $v_i\in x_i.f(X_{x_i}^r)$, we must have that $w_i\in x_i.X_{x_i}^r$. Thus $x_{j+1}$ lies at distance at most $r+1$ of $u_j. X_{u_j}^r$. Hence $x_j$ lies at distance at most $r+1$ of $\bigcup_{i=0}^{m-1} u_i.X_{u_i}^r\subseteq X^{2m(r+1)-1}$. Hence $x_j\in X^{2m(r+1)+r}$, and thus $\bigcup_{i=0\ldots m} x_i.X_{x_i}^r \subseteq X^{2m(r+1)+2r+1}$. Hence $X^{2(m+1)(r+1)-1}$ determines the internal states of $H(X)_\varepsilon^m$.\\
Summarizing, $X^{n}$, with $n=2(m+1)(r+1)-1$ determines $F(X)^m=\tilde{H}(X)_\varepsilon^m$.\\
Consider some $v''\in R_X^m$. This means that $v''\in (RH(X))_\varepsilon^m$ and $v''=R(v')$ for some $v'\in H(X)$ that contains $v\in X$ in its name. Hence $v'\in H(X)_\varepsilon^m$, where we used $R(\varepsilon)=\varepsilon$. Since this determined by $X^n$, we have $v\in X^n$. Hence $\dom R_X^m\subseteq X^n$. Moreover, consider $X'$ such that $X'^r=X^r$. Therefore $v\in X'^r$, $H(X)_\varepsilon^m$ and $H(X')_\varepsilon^m$ are isomorphic, and this isomorphism sends $v'$ to the $w'$ of $H(X')_\varepsilon^m$ whose name contains $v$. 
Therefore $F(X)^m$ and $F(X')^m$ are equal, and the same paths designate $R_X^m(v)$ and $R_{X'}^m(v)$, which are thus equal.\\
{\bf $[$Caus.$\Rightarrow$Loc.$]$} Let $(F,R_\bullet)$ be a causal dynamics.  Let $b_0$ and $b_1$ be respectively the bounds given by Definition~\ref{def:boundednessmodulo} and Lemma~\ref{lem:connectivitymodulo}, and $b=\max(b_0+1,b_1)$.  Let $m=3b+2$. Let $r$ be the radius such that for all $X,X'$, $X^{r}=X'^{r}$ implies $F(X)^m=F(X')^m$ and $R_X^m=R_{X'}^m$, from Definition \ref{def:causal} and Heine's Theorem. 
We will construct $f$ from ${\cal X}^r$ to ${\cal G}_{\Sigma, \Delta, \ports}$ so that for all $X^r$, the graph $f(X^r)$ is a well-chosen member of the equivalence class $F(X^r)^b$. Hence we must instantiate $F(X^r)^b$ via a suitable, local naming of its vertices. We use the isomorphism $S_{X^r}$ of Lemma \ref{lem:locrenaming} for this purpose, i.e. $f(X^r)=S_{X^r}G(F(X^r)^b)$.\\
$[$Dynamics$]$ For all $X^r$, $f(X^r)$ has vertices that are subsets of $V(X^r).S$, by definition. These sets are disjoint, by Lemma \ref{lem:locrenaming} $(i)$ applied to pairs of vertices of $F(X^r)^b$. Moreover $\varepsilon\in f(X^r)$, since $\varepsilon\in F(X^r)^b$ and $S_{X^r}(\varepsilon)=\varepsilon$ by Lemma \ref{lem:locrenaming} $(ii)$.\\
$[$Boundedness$]$ For all $X^r$, the graph $f(X^r)$ is finite, by construction.\\
$[$Consistency$]$ In order to show the consistency of $f$, we will show that for all $X$, $u\in X$, we have that $u.f(X_u^r)$ is the subgraph $H(X)_u^b$ of $H(X)$, where $H(X)$ is a well-chosen member of the equivalence class $F(X)$. Hence we must instantiate $F(X)$ via a suitable naming of its vertices. We use the isomorphism $S_X$ of Lemma \ref{lem:locrenaming} for this purpose, i.e. $H(X)=S_X G(F(X))$.\\
Start from $u.f(X_u^r)=u.S_{X_u^r} G(F(X_u^r)^b)$, which is equal to $u.S_{X_u} G(F(X_u)^b)$, by Lemma \ref{lem:locrenaming} $(iii)$ and using the fact that $F(Y)^b=F(Y^r)^b$. This, in turn, is equal to $u.\big(S_{X_u} G(F(X_u))\big)^b$, using the natural inclusion of $F(Y)^b$ into $F(Y)$. This, in turn, is equal to $u.\big(S_{X_u} \overline{R_X(u)}.G(F(X))\big)^b$, by shift-invariance, which is equal to $u.\big(\overline{u}.S_X G(F(X))\big)^b$, by Lemma \ref{lem:locrenaming} $(iv)$. This, finally, is $\big(S_X G(F(X))\big)_u^b=H(X)_u^b$, since it is true that for any graph $G$ and any isomorphism $T$, $TG_u^b=(TG)_{T(u)}^b$ and thus $G_u^b=T^{-1}(TG)_{T(u)}^b$.\\
Summarizing, $u.f(X_u^r)=H(X)_u^b$. Moreover if $u\in X^0$, then notice that $u\in f(X^r)$ and $u\in u.f(X_u^r)$, and hence they are non-trivially consistent.\\
Since $f$ is consistent, and $f(X^r)$ is a representant of $F(X^r)^b$, it remains only to remark that $F(X)=\bigcup u.F(X_u^r)^b$, which is true because $b$ was chosen to be strictly larger the one given by Definition~\ref{def:boundednessmodulo}, insuring that all the vertices and edges of $F(X^r)$ are covered, along with their labels.\qed
\end{proof} }

In the proof of Theorem \ref{th:main}, the renaming $S_X$ takes a generalized Cayley graph $F(X)$ into a mere graph $H(X)$. It does so by providing names for the vertices of $F(X)$, that are subsets of $V(X).S$. The idea is that $w'$ in $F(X)$ gets named $S_X(w')$, which is the set of those $u.z$, such that $u'=R_X(u)$ is close to $w'$, and $z$ is an integer encoding the remaining path between $u'$ and $w'$. The following lemma formalizes this idea as well as some useful, technical although expected properties.
\begin{lemma}[Local renaming properties]\label{lem:locrenaming}
Let $(F,R_\bullet)$ be a causal dynamics.  Let $b$ be the maximum of the bounds from Definition \ref{def:boundednessmodulo} and Lemma \ref{lem:connectivitymodulo}.  Let $m=3b+2$.  Let $r$ be the radius such that for all $X,X'$, $X^{r}=X'^{r}$ implies $F(X)^{m}=F(X')^{m}$ and $R_X^m=R_{X'}^m$, from Definition \ref{def:causal} and Heine's Theorem. Let $z$ be an injection from $V({\cal X}^b_\pi)\backslash\varepsilon$, as in Definition \ref{def:size}, to $\mathbb{N}$. Let $z(\varepsilon)$ be the empty word. Let $Y$ be a generalized Cayley graph. Consider $S_Y$ such that for all $w'\in F(Y)$ we have 
$$S_{Y}(w')=\{u.z(v')\,|\,u'.v'=w'\,\wedge\,u\in Y\,\wedge\,u'=R_{Y}(u) \,\wedge\,v'\in F(Y)_{u'}^b \}.$$
We have:
\begin{itemize}
\item[(i)] $\forall w_1', w_2'\in F(Y)$, $S_{Y}(w_1')\cap S_{Y}(w_2')\neq \emptyset \Rightarrow S_{Y}(w_1')=S_{Y}(w_2')$.
\item[(ii)] $\varepsilon\in S_Y(\varepsilon)$.
\item[(iii)] $\forall w'\in F(X_u)^b$, $u.S_{X_u^r}(w') = u.S_{X_u}(w')$.
\item[(iv)] $\forall v'\in F(X_u)$, $S_X(R_X(u).v')=u.S_{X_u}(v')$.
\end{itemize}
\end{lemma}
\begin{proof}
$[(i)]$ Consider $w_1', w_2'$ such that $S_{Y}(w_1')$ and $S_{Y}(w_2')$ have a common element $u.z(v')$. This entails that  $w_1'=u'.v'=w_2'$ is the same vertex in $F(Y)$, and thus that $S_{Y}(w_1')=S_{Y}(w_2')$.\\
$[(ii)]$ Since $z(\varepsilon)=\varepsilon$, $\varepsilon.\varepsilon=\varepsilon$, $\varepsilon=R_Y(\varepsilon)$ and $\varepsilon\in F(Y)^b$. \\
$[(iii)]$ Consider the $u=\varepsilon$ case. Let $w'$ be a vertex of $F(X)^b$, and $u'\in F(X)$ a vertex such that $u'.v'=w'$, with $|v'|\leq b+1$. We necessarily have that $u'\in F(X)^{2b+1}$. Moreover, since $F(X)^{3b+2}=F(X^r)^{3b+2}$, we have $F(X)_{u'}^b=F(X^r)_{u'}^b$. Also, using $R_{X}^m=R_{X^r}^m$, we have that
$$u'=R_{X}(u)\,\Leftrightarrow\, u'=R_{X}^m(u)\,\Leftrightarrow\, u'=R_{X^r}^m(u)\,\Leftrightarrow\, u'=R_{X^r}(u).$$
where the middle equivalence uses the natural inclusion of $X^r$ into $X$. As a consequence the two sets:
\begin{align*}
S_{X}(w')&=\{u.z(v')\,|\,u'.v'=w'\,\wedge\,u\in X\,\wedge\,u'=R_{X}(u) \,\wedge\,v'\in F(X)_{u'}^b \}\\
S_{X^r}(w')&=\{u.z(v')\,|\,u'.v'=w'\,\wedge\,u\in X^r\,\wedge\,u'=R_{X^r}(u) \,\wedge\,v'\in F(X^r)_{u'}^b \}
\end{align*}
are equal, up to the natural inclusion of $X^r$ into $X$. The same holds for $S_{X_u}$ and $S_{X_u^r}$. Then, since the shift operation $(u.)$ is from $V(X^n)$ to $V(X)$, full equality holds between $u.S_{X_u}$ and $u.S_{X_u^r}$.\\
$[(iv)]$ Consider some $u'.v'.w'\in F(X)$ with $u'=R_{X}(u)$, $v'=R_{X_u}(v)$ and $w'\in F(X_{u.v})^b$.
\begin{align*}
u.S_{X_u}(v'.w') &= u.\{x.z(y')\,|\,v'.w'=x'.y'\,\wedge\,x\in X_u\\
& \hspace{3cm} \wedge\,x'=R_{X_u}(x) \,\wedge\,y'\in F(Y)_u'^b\}\\
&= \{u.x.z(y')\,|\,u'.v'.w'=u'.x'.y'\,\wedge\,u.x\in X\\
&\hspace{3cm} \wedge u'.x'=R_{X}(u.x) \,\wedge\,y'\in F(Y)_u'^b\}\\
&= S_X(u'.v'.w')\\
&= S_X (R_X(u).v'.w')
\end{align*}\qed
\end{proof}

Our causal dynamics over generalized Cayley graphs is a candidate model of computation accounting for space, but without this space being fixed. As a candidate model of computation, we must check that it is computable. The following shows that we can decide whether a syntactic object is a valid instance of the model.
\begin{proposition}[Decidability of consistency]\label{prop:decidable}
Given a dynamics $f$ from ${\cal X}^r_{\Sigma, \Delta, \ports}$ to ${\cal G}_{\Sigma, \Delta, \ports}$, it is decidable whether $f$ is a local rule.
\end{proposition}
\begin{proof}
First of all notice that there is a finite number of disks $X^b$  of radius $b$, with labels in finite sets $\Delta$ and $\Sigma$. 
The following informal procedure verifies that $f$ is a local rule:
\begin{itemize}
\item[$\bullet$] For each $X^r$ check that $\varepsilon \in f(X^r)$.
\item[$\bullet$] For each $X^{r+1}$ check that for all $u\in X^0$, $f(X^r)$ and $u.f(X^r_u)$ are non-trivially consistent.
\item[$\bullet$] For each $X^{3r+2}$ check that for all $u\in X^{2r+1}$, $f(X^r)$ and $u.f(X^r_u)$ are non-trivially consistent.
\end{itemize}\qed
\end{proof}

\noindent Finally, we prove that if the initial state is finite, its evolution can be computed.
\begin{proposition}[Computability of causal functions]\label{prop:computable}
Given a local rule $f$ and a finite generalized Cayley graph $X$, then $F(X)$ is computable, with $F$ the causal dynamics induced by $f$.
\end{proposition}
\begin{proof}
Since $f$ is a local rule, the images of disks of radius $r$ included in $X$ are all finite, and consistent with one another. Moreover the finite union of finite, consistent graphs, is computable.\qed
\end{proof}

\section{Properties}\label{sec:Consequences}

\noindent {\em Composability.} We have characterized causal dynamics as the continuous, shift-invariant, bounded functions over generalized Cayley graphs. An important question is whether this notion is general enough. A good indicator of this robustness is that it is stable under composition.
\begin{definition}[Composition]
Consider two dynamics $(F,R_\bullet)$ and $(G,S_\bullet)$. Their composition $(G,S_\bullet)\circ (F,R_\bullet)$ is $(G\circ F,T_\bullet)$ where $T_X=S_{F(X)}\circ R_X$, i.e. $T_X(v)= S_{F(X)}(R_X(v))$.
\end{definition}
Indeed, stability under composition holds for classical and reversible cellular automata, but has failed to be obtained for the early definitions of probabilistic cellular automata and quantum cellular automata (see \cite{ArrighiPCA} and \cite{DurrUnitary,SchumacherWerner,ArrighiLATA} for a discussion).

\begin{theorem}[Composability]\cite{ArrighiCGD}\label{th:composability}
Consider causal dynamics $(F,R_{\bullet})$ and $(G,S_{\bullet})$, both over ${\cal X}_{\Sigma,\Delta,\ports}$. Then their composition is also a causal dynamics.
\end{theorem}
\proof{
$[$Continuous$]$ In the $F', G':{\cal X}_{\Sigma',\Delta,\ports}\to{\cal X}_{\Sigma',\Delta,\ports}$ formalism, it suffices to state that the composition of two continuous functions is continuous. Without this formalism this decomposes into:
\begin{itemize} 
\item[$\bullet$] $(G\circ F)$ is continuous because it is the composition of two continuous functions.
\item[$\bullet$] Consider $T_{\bullet}=S_{F(\bullet)} \circ R_{\bullet}$. For all $X$, for all $m$, there exists $n$ such that for all $X'$, $X'^n=X^n$ implies $T_{X'}^m=T_{X}^m$. Indeed:
\end{itemize}
Fix some $X$ and $m$. Since $(G,S_{\bullet})$ is a causal dynamics, there exists $n'$ such that for all $X'$, $F(X')^{n'}=D'=F(X)^{n'}$ implies $S_{F(X')}^m=S_{D'}^m=S_{F(X)}^m$. Fix this $n'$. Since $(F,R_{\bullet})$ is a causal dynamics, there exists $n$ a radius such that for all $X'$, $X^{n}=D=X'^{n}$ implies $F(X)^{n'}=F(X')^{n'}$ and $R_X^{n'}=R_D^{n'}=R_{X'}^{n'}$. Now, for this $n$, $T_{X'}^m=S_{F(X')}^m\circ R_{X'}^{n'}=S_{D'}^m\circ R_{X'}^{n'}=S_{D'}^m\circ R_D^{n'}$, which, by the symmetrical is equal to $T_{X}^m$.\\
$[$Shift-invariant$]$ We have $G(F(X_u))=G(F(X)_{R_X(u)})=G(F(X))_{S_{F(X)}(R_X(u))}$,
$T_X(u.v)=S_{F(X)}(R_X(u.v))=S_{F(X)}(R_X(u).R_{X_u}(v)))=S_{F(X)}(R_X(u)).S_{F(X)_{R_X(u)}}(R_{X_u}(v))=T_X(u).S_{F(X_u)}(R_{X_u}(v))=T_X(u).T_{X_u}(v)$.\\
$[$Bounded$]$ Since $(G,S_{\bullet})$ is a causal dynamics, there exists a bound $b''$ such that for all $X$, for all $w''\in G(F(X))$, there exists $x''=S_{F(X)}(x')$ and $v''\in G(F(X))_{x''}^{b''}$ 
such that $w''=x''.v''$. Since $(F,R_{\bullet})$ is a causal dynamics, there exists a bound $b'$ such that there exists $u'=S_{F(X)}(u)$ and $v'\in F(X)_{u'}^{b''}$ such that $x'=u'.v'$. Let $u''=S_{F(X)}(u')=S_{F(X)}(R_X(u))=T_X(u)$. Now, according to Lemma \ref{lem:connectivitymodulo} applied to $(G,S_{\bullet})$ and points $u'$ and $x'$, there exists a bound $c$ such that there exists $t''\in G(F(X))_{u''}^{c.(b'+1)}$ and $x''=u''.t''$. Let $b=c.(b'+1)+b''$, we now have that for $u''=S_{F(X)}(u')=S_{F(X)}(R_X(u))=T_X(u)$ there exists $v''.t''\in G(F(X))_{u''}^{b}$ such that $w''=u''.t''.v''$. \qed}

The above proof was done via the axiomatic characterization of causal dynamics, as this paper enjoys a more straightforward formalization than \cite{ArrighiCGD}. In \cite{ArrighiCGD} the same result is proven via the constructive approach to causal graph dynamics (localizability), which has the advantage of extra information about the composed function. It establishes the following. Consider $F$ a causal dynamics induced by the local rule $f$ of radius $r$ (i.e. diameter $d=2r+1$). Consider $G$ a causal graph dynamics induced by the local rule $g$ of radius $s$ (i.e. diameter $e=2s+1$). Then $G\circ F$ is a causal graph dynamics induced by the local rule $g$ of radius $t=2rs+r+s$ (i.e. diameter $f=de$) from ${\cal D}^{t}$ to ${\cal G}_{\Sigma, \Delta, \pi}$ which maps $X^t$ to
$$\bigcup_{v\in X'} v.g(X'^{s}_{v})\quad\textrm{with}\quad X'=\bigcup_{u \in X^t} u.f(X^{r}_u).$$
The same result, with the transposed proof, still holds.

\noindent {\em Invertibility implies reversibility.} Let us turn our attention to some set-theoretical properties.

\begin{definition}[Shift-invariant invertible]
A shift-invariant dynamics $(F,R_\bullet)$ is {\em shift-invariant invertible} if and only if $F$ is a bijection, and there is an $S_\bullet$ such that $(F^{-1},S_\bullet)$ is a shift-invariant dynamics.
\end{definition}
Notice that there exists some shift-invariant dynamics $(F,R_\bullet)$ such that $F$ is a bijection but there exists no $S_\bullet$ such that $(F^{-1},S_\bullet)$ is a shift-invariant dynamics (Outline: take $\pi=\{a,b\}$, and map the $4$-sized directed segment to the $7$-sized directed segment pointed on the first four positions, and the $3$-sized directed segment to the $7$-sized directed segment but pointed on the last three positions). In this paper, we will not consider them. Notice also that there exists some shift-invariant dynamics $(F,R_\bullet)$ such that $(F^{-1},S_\bullet)$ is a shift-invariant dynamics, but $S_{F(X)}$ is not the inverse of $R_X$. The {\em Turtle example} of Figure \ref{fig:turtle} illustrates this possibility. 
\begin{figure}[htpb]
{\centering
\includegraphics[scale=2]{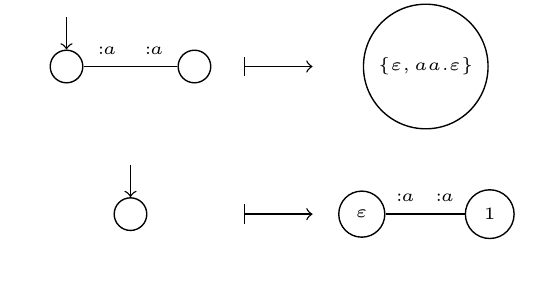}}
\caption{{\em The ``turtle'' dynamics.} \label{fig:turtle}}
\end{figure}
Again, in this paper, we will not consider them: we restrict ourselves to vertex-preserving invertible dynamics. 
\begin{definition}[Vertex-preserving invertible]
A shift-invariant dynamics $(F,R_\bullet)$ is {\em vertex-preserving invertible} if and only if $F$ is a bijection and for all $X$ we have that $R_X$ is a bijection.
\end{definition}
Those are automatically shift-invariant invertible:
\begin{lemma}[Vertex-preserving invertible is shift-invariant invertible]
If  $(F,R_\bullet)$ is a vertex-preserving invertible shift-invariant dynamics, then $(F^{-1},S_\bullet)$ is a shift-invariant dynamics, with $S_{Y}=(R_{F^{-1}(Y)})^{-1}$.
\end{lemma}
\proof{
Consider $Y$ and $u'.v'\in Y$. Take $X$ and $u.v\in X$ such that $F(X)=Y$, $R_X(u)=u'$ and $R_X(u.v)=u'.v'$. We have:
$$ F^{-1}(Y_{u'})=F^{-1}(F(X)_{R_X(u)})=F^{-1}(F(X_u))=X_{(R_X)^{-1}(u')}=F^{-1}(Y)_{S_Y(u')}.$$
Moreover, take $v\in X_u$ such that $R_X(u.v)=R_X(u).R_{X_u}(v)=u'.v'$. We have: 
$$ S_Y(u'.v')=(R_X)^{-1}(R_X(u.v))=u.v=(R_{X})^{-1}(u').(R_{X_u})^{-1}(v')= S_Y(u').S_{Y_{u'}}(v').$$}
Back to the case of a causal dynamics, the classical question to ask is whether the inverse is also a causal dynamics. 
\begin{definition}[Reversible]
A causal dynamics $(F,R_{\bullet})$ is {\em reversible} if and only if it is shift-invariant invertible with inverse $(F^{-1} ,S_{\bullet})$ a causal dynamics.
\end{definition}
The big question is whether the causality of a forward-time causal evolution $F$, entails that of the backward-time evolution $F^{-1}$. In other words: is causality stable under inversion? This question was answered positively in the earlier formalism of \cite{ArrighiCGD}, with a more lengthy proof.

\begin{theorem}[Reversibility]\label{th:reversibility}
Consider $(F,R_\bullet)$ a causal dynamics. If $(F,R_\bullet)$ is vertex-preserving invertible, then it is reversible, with inverse $(F^{-1},S_\bullet)$ where $S_{Y}=(R_{F^{-1}(Y)})^{-1}$.
\end{theorem}
\proof{For this proof it is convenient to switch to the the $F':{\cal X}_{\Sigma',\Delta,\ports}\to{\cal X}_{\Sigma',\Delta,\ports}$ formalism, introduced right after Def. \ref{def:dynamicsmodulo}. Since $F'=(F,R_\bullet)$ is shift-invariant invertible we have that $F'^{-1}=(F^{-1},S_\bullet)$ is shift-invariant. Since $F'$ is continuous over the compact space ${\cal X}_{\Sigma',\Delta,\ports}$, with $\Sigma'=\Sigma\times\{0,1\}$, we have that $F'^{-1}=(F^{-1},S_\bullet)$ is continuous. Since $R_X$ is bijective, so is $S_{F(X)}$, and thus so is $S_Y$ for any $Y$. Hence, $S$ is surjective and so $(F^{-1},S_\bullet)$ is bounded with bound $0$.}
\medskip

\noindent {\em Discussion: Garden-of-Eden.} 
Another important result in Cellular Automata theory, and which is related to invertibility questions, is the so-called Garden-of-Eden (a.k.a Moore-Myhill theorem), which states that pre-injectivity (i.e. injectivity over the set of finite configurations) is equivalent to surjectivity (over the set of configurations). This result has been extended to Cellular Automata over Cayley graphs, provided that the group which induces the Cayley graph has a certain property (it must be amenable) see \cite{CeccheriniEden}. Extending this result to a wider class of graphs is impossible for the surjective implies pre-injective part \cite{}, and is the subject of ongoing research for the pre-injective implies surjective part, see for instance \cite{Gromov}.\\
In the setting of this paper, there are at least two good reasons for the Garden-of-Eden theorem {\em not} to hold. The first reason is that here, Cellular Automata have been extended to generalized Cayley graphs, encompassing not just amenable Cayley graphs, but also the non-amenable ones, and may others: in fact all arbitrary finite degree graphs. The second reason is that here, Cellular Automata have been extended to time-varying graphs, for which pre-injectivity becomes a much weaker constraint (counting arguments fail as injectivity can be maintained by generating extra vertices, instead of saturating the space of internal states). For instance a causal dynamics which just adds a vertex to every free port, is injective but not surjective. It could be interesting, however, to look for non-trivial subclasses of causal dynamics for which the Garden-of-Eden still holds. 

\noindent {\em Discussion: Subclasses of causal dynamics.}

Finally we mention two natural subclasses of graph dynamics. The first is that where only the topology of the graph is evolving, i.e. there is no internal state on vertices nor edges. 
\begin{definition}[No-state a.k.a graph-only dynamics]
A dynamics $(F,R_{\bullet})$ is a {\em graph-only a.k.a no-state dynamics} if and only if it is defined over ${\cal X}_\pi$, i.e. the graphs carry no internal state.
\end{definition}
All of our results apply unchanged in this graph-only setting, as there was nowhere a particular need for an internal state. Moreover, it seems clear that causal graph-only dynamics can simulate general causal dynamics elegantly. But it is not so clear whether this still holds in the reversible case, for instance.\\
The second, dual class is that where only the internal states are evolving, i.e. the dynamics does not change the graph. This is of course is a widely studied case  \cite{PapazianRemila,Gruner,Gruner2,Gromov,CeccheriniEden,TomitaGRA,KreowskiKuske,TomitaSelfReproduction,TomitaSelfDescription,BFHAmalgamation,LoweAlgebraic,EhrigLowe,Taentzer,TaentzerHL}.
\begin{definition}[Graph-preserving a.k.a state-only dynamics]
A dynamics $(F,R_{\bullet})$ is a {\em graph-preserving a.k.a state-only dynamics} if and only if for all $X$, the graphs $X$ and $F(X)$ have the same structure, i.e. they are the same up to labellings $\sigma,\delta$.
\end{definition}
Again all of our results apply unchanged in this state-only setting, as there was nowhere a particular need for changing the topology, although changing the topology is one of the main contributions of this paper. Still, it could be said that the paper does port the Curtis-Hedlund-Lyndon theorem to Cellular Automata over arbitrary graphs, and not just Cayley graphs, which had not been done. Some results could of course be made tighter for state-only causal dynamics, such as that of the radius of a composition.\\
This graph-preserving class was defined by demanding that a certain property be preserved by the graph dynamics. Thus it falls into the broad class of subspace-preserving dynamics:
\begin{definition}[Subspace-preserving dynamics]
Consider $\{{\cal X}_1,\ldots,{\cal X}_n\}$ a partition of ${\cal X}_{\Sigma',\Delta,\ports}$. A dynamics $(F,R_{\bullet})$ is a {\em subspace-preserving dynamics} with respect to the partition if and only if for all $i$, we have that $X\in{\cal X}_i$ implies that $F(X)\in{\cal X}_i$.
\end{definition}
On the other hand, the no-state class was defined by restricting the definition space of the causal dynamics. Thus it falls into the broad class of subspace-restricted dynamics:
\begin{definition}[Subspace-restricted dynamics]
Consider ${\cal Y}$ a subset of ${\cal X}_{\Sigma',\Delta,\ports}$. A dynamics $(F,R_{\bullet})$ is a {\em subspace-restricted dynamics} with respect to ${\cal Y}$ if and only if its definition is restricted to ${\cal Y}$.
\end{definition}
In both of these broad classes, it seems cautious to demand that the subspaces be themselves compact spaces, as was the case with the no-state and graph-preserving classes. Finally, let us mention that in our study of reversibility, we required that our causal dynamics $(F,R_{\bullet})$ be vertex-preserving, i.e. that $R_X$ be a bijection between $V(X)$ and $V(F(X))$. This differs from graph-preservation: the connectivity may vary. Is is not clear whether this vertex-preserving class could have been defined through a subspace-preservation construction. 

\section{Conclusion}

{\em Summary.} First we have shown that a notion of graphs with ports modulo isomorphism (Definitions \ref{def:graphs}--\ref{def:pointedmodulo}) provides a generalization of Cayley graphs, in the following sense: each vertex can be named relatively to the origin; each graph represents a language and its equivalence relation (Definitions \ref{def:completeness}--\ref{def:adjacencystructures}, Theorem \ref{th:SP}); and they are equipped with a well-defined notion of translation (Definition \ref{def:shift}).
Second, we have shown that the set of these graphs forms a compact metric space (Definition \ref{def:metric} and Lemma \ref{lem:compactness}), entailing that continuous functions over this set are also uniformly continuous (Heine's theorem).
Third, this allowed us to characterize Cellular Automata over those generalized Cayley graphs as the set of shift-invariant, continuous, bounded dynamics (Definitions \ref{def:dynamicsmodulo}-\ref{def:causal}). This physically-motivated mathematical definition would have remained excessively abstract without our main result, showing that such causal dynamics are necessarily localizable, i.e. that they can be expressed as the synchronous, homogeneous application of a local rule (Definitions \ref{def:dynamics}-\ref{def:localizable}, Theorem \ref{th:main}). Fourth, we showed that the property of being a local rule is decidable and hence that causal dynamics are computable (Propositions \ref{prop:decidable}-\ref{prop:computable}).  Finally, we showed that the composition of two causal dynamics is itself a causal dynamics (Theorem \ref{th:composability}), and that the shift-invariant inverse of a causal dynamics is again a causal dynamics (Theorem \ref{th:reversibility}).

{\em Further works.} The mathematical relation between the causal dynamics of \cite{ArrighiCGD} and ours remains to be clarified -- for instance, decidability remains to be proven for the causal graph dynamics of \cite{ArrighiCGD}. Still, they are important features of models of computation. The fact that they are relatively straightforward to prove in this paper is a good indicator that the formalism presented is appropriate.\\ 
Our short terms plan, however, include: interpreting CA over generalized Cayley graphs as a dynamics over simplicial complexes as was started in \cite{ArrighiSURFACES}; deepening the study of the reversible case; formalizing the stochastic case. 
Moreover, one of the authors has been studying the intrinsic simulation and intrinsic universality of causal graph dynamics in \cite{MarMar}, an approach which  can still be taken further.

\section*{Acknowledgements} This work has been funded by the ANR-10-JCJC-0208 CausaQ grant and by the John Templeton Foundation, grant ID 15619.  It also benefited from discussions with Christophe Crespelle, Gilles Dowek, Emmanuel Jeandel, Viv Kendon, Jean Mairesse, Bruno Martin, David Meyer, Simon Perdrix, and \'Eric Thierry. We thank the anonymous referees, who pushed for a better paper.

\bibliography{biblio}
\bibliographystyle{plain}

\end{document}